\documentclass[a4paper,USenglish]{lipics-v2018}

\usepackage{color}
\usepackage{xspace}
\usepackage{marvosym}
\usepackage{enumitem}
\usepackage{booktabs}
\usepackage{multirow}
\usepackage[noend]{algpseudocode}

\nolinenumbers 
\hideLIPIcs 

\newcommand{\specialcell}[2][c]{\begin{tabular}[#1]{@{}c@{}}#2\end{tabular}}

\theoremstyle{plain}

\newtheorem{claim}{Claim}
\numberwithin{claim}{theorem}

\newcommand{\CountSHnoCDConst}{\textsc{Count-SH-noCD-Const}\xspace}
\newcommand{\CountSHnoCDHigh}{\textsc{Count-SH-noCD-High}\xspace}
\newcommand{\CountCenternoCDConst}{\textsc{Count-Desig-noCD-Const}\xspace}
\newcommand{\CountCenternoCDHigh}{\textsc{Count-Desig-noCD-High}\xspace}
\newcommand{\EstUpperSH}{\textsc{EstUpper-SH}\xspace}
\newcommand{\CountSHCDConst}{\textsc{Count-SH-CD-Const}\xspace}
\newcommand{\CountSHCDHigh}{\textsc{Count-SH-CD-High}\xspace}
\newcommand{\EstUpperCenter}{\textsc{EstUpper-Desig}\xspace}
\newcommand{\CountCenterCDConst}{\textsc{Count-Desig-CD-Const}\xspace}
\newcommand{\CountCenterCDHigh}{\textsc{Count-Desig-CD-High}\xspace}
\newcommand{\CountAllnoCDa}{\textsc{Count-All-noCD}\xspace}
\newcommand{\CountAllnoCDb}{\textsc{Count-All-noCD-2}\xspace}
\newcommand{\CountAllCDa}{\textsc{Count-All-CD}\xspace}

\title{Approximate Neighbor Counting in Radio Networks}
\author{Calvin Newport}{Georgetown University, United States}{cnewport@cs.georgetown.edu}{}{}
\author{Chaodong Zheng}{State Key Laboratory for Novel Software Technology, Nanjing University, China}{chaodong@nju.edu.cn}{}{}
\authorrunning{C. Newport and C. Zheng}
\Copyright{Calvin Newport and Chaodong Zheng}
\subjclass{}

\begin{document}

\maketitle

\begin{abstract}
For many distributed algorithms, neighborhood size is an important parameter. In radio networks, however, obtaining this information can be difficult due to ad hoc deployments  and communication that occurs on a collision-prone shared channel. This paper conducts a comprehensive survey of the approximate neighbor counting problem, which requires nodes to obtain a constant factor approximation of the size of their network neighborhood. We produce new lower and upper bounds for three main variations of this problem in the radio network model: (a) the network is single-hop and every node must obtain an estimate of its neighborhood size; (b) the network is multi-hop and only a designated node must obtain an estimate of its neighborhood size; and (c) the network is multi-hop and every node must obtain an estimate of its neighborhood size. In studying these problem variations, we consider solutions with and without collision detection, and with both constant and high success probability. Some of our results are extensions of existing strategies, while others require technical innovations. We argue this collection of results provides insight into the nature of this well-motivated problem (including how it differs from related symmetry breaking tasks in radio networks), and provides a useful toolbox for algorithm designers tackling higher level problems that might benefit from neighborhood size estimates.
\end{abstract}


\section{Introduction}\label{sec-intro}

Many distributed algorithms assume nodes have advance knowledge of their neighborhood, allowing them to take steps that depend, for example, on gathering information from every neighbor (e.g.,~\cite{luby85}), or flipping a coin weighted with their neighborhood size (e.g.,~\cite{bar-yehuda87}).

In standard wired network models, where nodes are connected by static point-to-point links, obtaining this neighbor information is often trivial (e.g., as in the LOCAL or CONGEST models). In radio networks, by contrast, this information might be harder to obtain. Specifically, because nodes in these networks are often deployed in an ad hoc manner, and subsequently communicate only on a contended shared channel, we cannot expect that they possess advance knowledge of their neighborhood. In fact, learning this information might require non-trivial feats of contention management.

Some distributed algorithms for radio networks depend on nodes possessing an estimate of their neighborhood size (e.g.,~\cite{gilbert14,gilbert17}), while other algorithms could be significantly simplified if this information was available (e.g.,~\cite{willard86,nakano02,jurdzinski05}). Though it is generally assumed that calculating these size estimates should not take {\em too} long in most settings, this problem has escaped the more systematic scrutiny applied to related tasks like contention resolution.

In this paper, we work toward filling in more of this knowledge gap. We conduct a comprehensive survey of lower and upper bounds for the approximate neighbor counting problem in the radio network model under different combinations of common assumptions for this setting. Some of our results require only extensions of existing strategies, while many others require non-trivial technical innovations.

Combined, this collection of results provides two important contributions to the study of distributed algorithms for radio networks. First, it supports a deeper understanding of the well-motivated neighbor counting problem, highlighting both its similarities and differences to related low-level radio network tasks. Second, the collection acts as a useful toolbox for algorithm designers tackling higher level problems.

\subparagraph*{Result summary.} The radio network model we study describes the underlying network topology with an undirected connected graph $G=(V,E)$, with the $n=|V|$ vertices corresponding to the radio devices (usually called {\em nodes} in this paper), and the edges in $E$ describing which node pairs are within communication range. For every node $u\in V$, $n_u$ describes the number of neighbors of $u$ in $G$. We sometimes call this parameter the {\em neighbor count} of $u$. In single-hop networks (i.e., $G$ is a clique), all nodes have the same neighbor count, while in multi-hop networks these counts can differ.

The {\em approximate neighbor counting} problem requires nodes to calculate constant factor estimates of their neighbor counts. We study the variant where every node must obtain this estimate (e.g., during network initialization), and the variant where only a designated node must obtain this estimate (e.g., when neighborhood of a node changes). We study these variants in single-hop and multi-hop networks, and consider solutions with and without collision detection. We study both lower and upper bounds for randomized solutions. When relevant, we look at both results that hold with constant and high probability.

Our results are summarized in Figure~\ref{tbl-results}. Notice that we do not study both designated and all nodes counting in single-hop networks, as in this setting all nodes have the same neighbor count, making these two cases essentially identical (e.g., a designated node in a single-hop network can simply announce its count, transforming the solution to an all nodes counting solution). We also do not study constant probability solutions for all nodes counting in multi-hop networks. This follows because in the multi-hop setting the success probability applies to each individual node. A constant success probability, therefore, implies that a constant fraction of the nodes are expected to generate inaccurate neighbor counts---a result that is too weak in most scenarios. In the single-hop setting, by contrast, the success probability refers to the probability that {\em all} nodes generate good counts.

Also notice that two upper bounds are given for multi-hop all nodes counting without collision detection: $O(\lg^2{n_u})$ and $O(\lg^3{N})$. The first bound describes an algorithm that generates good neighbor counts but never terminates (specifically, each node must keep participating to help neighbors that are still counting). The second bound does terminate, but requires an upper bound $N$ on the maximum possible network size. This is the only algorithm we study that requires this information to work properly.

\begin{figure}[!t]
\centering
\begin{small}
\begin{tabular}{cccccc}
\toprule
& & \multicolumn{2}{c}{with constant probability} & \multicolumn{2}{c}{with high probability} \\
\cmidrule(lr){3-4}\cmidrule(lr){5-6}
& & no-CD & CD & no-CD & CD \\
\midrule
\multirow{2}{*}{\specialcell{all nodes in single-hop\\ (first variant)}} & lower bound & $\Omega(\lg{N})$ & $\Omega(\lg\lg{N})$ & $\Omega(\lg^2{N})$ & $\Omega(\lg{N})$ \\
\cmidrule(lr){2-6}
& upper bound & $O(\lg{n})$ & $O(\lg\lg{n})$ & $O(\lg^2{n})$ & $O(\lg{n})$ \\
\midrule
\multirow{2}{*}{\specialcell{designated node in multi-hop\\ (second variant)}} & lower bound & $\Omega(\lg{N_{\Delta}})$ & $\Omega(\lg\lg{N_{\Delta}})$ & $\Omega(\lg^2{N_{\Delta}})$ & $\Omega(\lg{N_{\Delta}})$ \\
\cmidrule(lr){2-6}
& upper bound & $O(\lg{n_{w}})$ & $O(\lg\lg{n_{w}})$ & $O(\lg^2{n_{w}})$ & $O(\lg{n_{w}})$ \\
\midrule
\multirow{2}{*}{\specialcell{all nodes in multi-hop\\ (third variant)}} & lower bound & --- & --- & $\Omega(\lg^2{N_\Delta})$ & $\Omega(\lg{N_\Delta})$ \\
\cmidrule(lr){2-6}
& upper bound & --- & --- & \specialcell{$O(\lg^2{n_{u}})$,\\ $O(\lg^3{N})$} & $O(\lg^2{n_{u}})$ \\
\bottomrule
\end{tabular}
\end{small}
\captionsetup{width=0.9\textwidth}
\caption{
Summary of all approximate neighbor counting results proved in this paper. In the above table, ``CD'' and ``no-CD'' denote ``with collision detection'' and ``without collision detection'' respectively, while ``high probability'' is expressed with respect to the parameter in the bound. The $N$ and $N_{\Delta}$ terms describe upper bounds on the maximum neighbor count in single-hop and multi-hop networks, respectively. Our lower bounds are expressed with respect to these maximum sizes, while our upper bounds are expressed with respect to the actual network sizes, with the exception of the $O(\lg^3{N})$ bound for multi-hop all nodes counting without collision detection. This is the only algorithm we study that requires knowledge of network statistics to work properly.
}\label{tbl-results}
\vspace{-2ex}
\end{figure}

\subparagraph*{Discussion.} For all but one cases that we have lower bounds, they match our upper bounds. For the single-hop results, these bounds also match the relevant bounds from the related single-hop contention resolution problem (c.f.,~\cite{newport14}). In fact, most of the lower bounds in this single-hop setting follow by reduction from contention resolution. That is, we show that if you can solve approximate neighbor counting fast, then you can also solve contention resolution fast---allowing existing lower bounds from the latter to carry over to the former.

The single-hop upper bounds, however, required more than the simple application of existing contention resolution strategies. In contention resolution, for example, if you get lucky with your coin flips, and a node broadcasts alone earlier than expected, this is good news---you have solved the problem even faster! In neighbor counting, however, this ``luck'' might lead you to output an inaccurate size estimate. The analysis used for neighbor counting must bound the probabilities of these precocious symmetry breaking events.

Another complexity of neighbor counting (in single-hop networks) as compared to contention resolution is that {\em all} nodes must learn an estimate. This requires extra mechanisms to ensure that once some nodes learn a good estimate, this information is spread to all others. The most difficult single-hop case is the combination of high probability correctness and collision detection. To achieve an accurate estimate in an optimal $O(\lg{n})$ rounds required the adaptation of a technique based on one-dimensional random walks~\cite{nakano02,brandes17}.

Obtaining lower bounds for the multi-hop designated node setting required technical innovations. In the single-hop setting, our lower bounds used reduction arguments that applied the contention resolution bounds from~\cite{newport14} as a black box. In the multi-hop designated node setting, by contrast, we were forced to open the black boxes and modify them to handle the issues specific to multi-hop topologies. For the particular case of collision detection and high probability, substantial new arguments were needed to transform the bound.

In the multi-hop all nodes setting, obtaining upper bounds also required techniques beyond standard symmetry breaking strategies, as each node may simultaneously participate in multiple estimation processes. Our collision detector algorithm for this case has nodes use detectable noise to notify neighbors that they are still counting. When collision detection is not available, we consider two different approaches and hence present two algorithms. The first one returns an estimate for $n_u$ in $O(\lg^2{n_u})$ rounds, which is correct with high probability in $n_u$. The second algorithm uses a ``double counting'' trick and takes longer time, but the returned estimate is correct with high probability in $N$.

Last but not least, we would like to clarify a point about our lower bound statements. As shown in Figure~\ref{tbl-results}, our lower bounds are expressed with respect to the maximum possible neighbor counts (e.g., $N$ and $N_{\Delta}$), whereas, to obtain the strongest possible results, our upper bounds are expressed with respect to the actual neighbor counts in the analyzed execution (e.g., $n$ and $n_u$). The right way to interpret our lower bounds is that they claim in a setting where the number of participants comes from a set of $N$ (or $N_{\Delta}$) possible participants, there \emph{exists} a subset of these participants for which the stated bound holds.

Our lower bound technique does not directly tell us anything about the {\em size} of the participant set that induces the slow performance. Given our matching upper bounds, however, we can conclude that the worst case participant sets for these algorithms must have a size close to the maximum bounds. Consider, for example, single-hop counting with no collision detection. The lower bound says that for each algorithm there exists a collection of no more than $N$ participants that requires $\Omega(\lg{N})$ rounds to generate a good count with constant probability. Our upper bound, on the other hand, guarantees a good count in $O(\lg{n})$ rounds with constant probability, where $n$ is the size of the participant set. It follows that when the lower bound is applied to our algorithm, the bad participant set must have a size that is polynomial in $N$ (i.e.,  $n=\Theta(N^{\gamma})$, for some constant $0<\gamma\leq 1$), as otherwise the existence of both bounds is a logical contradiction.

\section{Related Work}

Algorithms to reduce contention and enable communication on shared channels date back to the early days of networking (c.f., \cite{greenberg87,cidon88}), and remain an active area of study today. In the study of distributed algorithms for shared {\em radio} channels, many strategies explicitly execute approximate neighbor counting as a subroutine. For example, in their study of energy-efficient initialization with collision detection, Bordim et al.~\cite{bordim99} propose a protocol that returns an estimate of $n$ in the range $[n/(16\lg{n}),2n/\lg{n}]$ within $O(\lg^2{n})$ time, while requiring each node to be awake for at most $O(\lg{n})$ rounds. Similarly, Gilbert et al.~\cite{gilbert17} use approximate neighbor counting as part of a neighbor discovery protocol in cognitive radio networks. It is also common for algorithms in this setting to simply assume these estimates are provided in advance. E.g., the often-used {\em decay} strategy introduced by Bar-Yehuda et al.~\cite{bar-yehuda87}, requires a bound on local neighborhood size to limit the estimates it tests.

As mentioned throughout this paper, neighbor counting is often closely related to contention resolution, which requires a single node to broadcast alone on the channel. Some common contention resolution strategies implicitly provide this approximation as a side-effect of their operation (e.g.,~\cite{willard86,nakano02,jurdzinski05}). At the same time, under some assumptions, a good estimate simplifies the problem of contention resolution. As we detail throughout this paper, however, this relationship is not exact. Lower bounds for neighbor counting often require more intricate arguments than contention resolution, and in some cases, contention resolution algorithms require nontrivial extra analysis and mechanisms to provide counts. Teasing apart this intertwined relationship is one of the main contributions of this paper.

Others have directly studied approximate neighbor counting in radio networks. Jurdzinski et al.~\cite{jurdzinski02} develop an algorithm that provides a constant factor approximation of $n$ within $O(\lg^{2+\delta}{n})$ time without collision detection for arbitrary constant $\delta>0$. Their algorithm guarantees that no node participates in more than $O((\lg\lg{n})^\delta)$ rounds. (Our relevant algorithm only needs $O(\lg^2{n})$ rounds, but consumes more energy.) Caragiannis et al.~\cite{caragiannis05} devise two constant-factor approximation algorithms: the first one requires collision detection and takes $O((\lg{n})\cdot(\lg\lg{n}))$ time, while the second one works without collision detection and takes $O(\lg^2{n})$ time. (Our relevant algorithms only need $O(\lg{n})$ rounds with collision detection, and perform as well as theirs without collision detection.) In \cite{kabarowski06,klonowski12}, the authors discuss how to approximate network size when adversaries are present.

Approximate neighbor counting has also been studied in the \textsc{beeping} model~\cite{cornejo10}, which is similar to, but somewhat weaker than, the standard radio network model. In this setting, Chen et al.~\cite{chen13} conduct an excellent mini survey on recent works in RFID counting (e.g., \cite{zheng11,shahzad12,zheng13,chen13}). They conclude that a two-phase approach is the key to achieve efficient and accurate RFID counting. They also prove several lower bounds, one of which shows $\Omega(\lg\lg{n})$ rounds are needed to obtain a constant factor approximation with constant probability. More recently, Brandes et al.~\cite{brandes17} study how to efficiently estimate the size of a single-hop \textsc{beeping} network: they provide both lower and upper bounds for a parameterized approximation accuracy. Notice, the main objective of \cite{chen13} and \cite{brandes17} differs from ours not just in the model, but in that they seek a $(1+\epsilon)$ approximation of $n$ for any $\epsilon>0$ ($\epsilon$ can be non-constant). Nonetheless, they both use constant factor approximation as a key subroutine.

\section{Model and Problem}\label{sec-model}

We consider a synchronous radio network. We model the topology of this network with a connected undirected graph $G=(V,E)$, with the $n=|V|$ vertices corresponding to the radio devices (usually called {\em nodes} in this paper), and the edges in $E$ describing which node pairs are within communication range.

For each node $u\in V$, we use $\Gamma_u$ to denote the set of neighbors of $u$, and use $n_u=|\Gamma_u|$ to denote the number of neighbors of $u$. Let $n_{\Delta} = \max_{u\in V}\{n_u\}$. Our algorithms assume $n_u \geq 1$. That is, we do not confront the possibility of a node isolated from the rest of a multi-hop network, or a single-hop network consisting of only a single node (we see the so-called {\em loneliness detection} problem as an interesting but somewhat orthogonal challenge; e.g.,~\cite{ghaffari12}). For the ease of presentation, we assume $n_u$ and $n$ are always a power of two. This assumption does not affect the correctness or asymptotic time complexities of our results. We define $N$ and $N_{\Delta}$ to be upper bounds on the maximum possible size of $n$ and $n_{\Delta}$, respectively. To obtain the strongest and most general possible results, our algorithms are {\em not} provided with knowledge of $N$ and $N_{\Delta}$, with the exception of an $O(\lg^3{N})$ time algorithm for multi-hop all nodes counting without collision detection. 

We divide time into discrete and synchronous {\em slots} that we also sometimes call {\em rounds}. We assume all nodes start execution during the same slot. (The definition of ``neighbor counting'' becomes complicated once nodes can activate in different time slots.) These assumptions imply nodes have access to a global clock. We assume each node is equipped with a half-duplex radio transceiver. That is, in each time slot, each node can choose to broadcast or listen, but cannot do both. If a node chooses to broadcast, then it gets no feedback from the communication channel. If a node chooses to listen and no neighbors of it broadcasts, then the node hears nothing (i.e., silence). If a node chooses to listen and exactly one of its neighbors broadcasts, then the node receives the message from that neighbor. Finally, if a node chooses to listen and at least two of its neighbors broadcast, then the result depends on the availability of a \emph{collision detection} mechanism: if collision detection is available, then the listening node hears noise; otherwise, the listening node hears nothing. As a result, without collision detection, a listening node cannot tell whether there are no neighbors broadcasting or there are multiple neighbors broadcasting.

In this paper, we are interested in the \emph{approximate neighbor counting} problem. This problem requires selected node(s) to obtain a constant factor approximation of their neighborhood size(s). In more detail, let constant $\tilde{c}\geq 1$ be the fixed approximation threshold for this problem. Each node $u$ that produces an estimate $\hat{n}_u$ must satisfy $n_u \leq \hat{n}_u \leq \tilde{c}\cdot n_u$. We consider three variations of this problem that differ with respect to the allowable network topologies and requirements on which nodes produce an estimate. The first variant assumes $G$ is single-hop and all nodes must produce an \emph{identical} estimate. The second variant assumes $G$ is multi-hop, but only a \emph{single} designated node $w$ must produce an estimate. The third variant is the same as the second, except that now {\em every} node must produce an estimate. We study randomized algorithms that are proved to be correct with a given probability $p$. In the single-hop variant, $p$ describes the probability of the event in which all nodes generate a single good approximation. In the multi-hop variants, by contrast, $p$ is the probability that an individual counting node generates a good approximation.

Throughout this paper, we cite the related {\em contention resolution} problem. In single-hop networks, the contention resolution problem is solved once some node broadcasts alone. Later in the paper, we consider a version of multi-hop contention resolution in which a single designated node must receive a message from a neighbor to solve the problem.

Finally, in the following, we say an event occurs {\em with high probability in parameter $k$} (or ``w.h.p.\ in $k$'') if it occurs with probability at least $1-1/k^\gamma$, for some constant $\gamma\geq 1$.

\section{Lower Bounds}

In this section, we presents our lower bounds for the approximate neighbor counting problem. We begin, in Section \ref{subsec-lower-bound-reduction} by looking at lower bounds that can be proved by reducing from the contention resolution problem. That is, in that subsection, we prove lower bounds by arguing that solving neighbor counting fast implies an efficient algorithm to contention resolution, allowing the relevant contention resolution lower bounds to apply.

We employ this approach to derive bounds for constant probability and high probability counting with no collision detection in both single-hop and designated node multi-hop settings. We also apply this approach to derive bounds for constant probability counting with collision detection in these settings. We cannot, however, apply this approach to high probability counting with collision detection, as the reduction itself is too slow compared to the desired bounds. We note that for the single-hop arguments, we leverage existing contention resolution bounds from \cite{newport14}. For the multi-hop arguments, however, we must first generalize the results from \cite{newport14} to hold for the considered network topology.

In Section \ref{subsec-lower-bound-new}, we look at lower bounds for high probability approximate neighbor counting with collision detection in both single-hop and designated node multi-hop settings. Unlike in Section \ref{subsec-lower-bound-reduction}, we cannot deploy a reduction-based argument. We instead prove a new lower bound that directly argues a sufficiently accurate estimate requires the stated rounds.

Finally, in Section \ref{subsec-lower-bound-carry-over} we look at lower bounds for the remaining case of multi-hop all nodes counting. We establish these bounds by reduction from designated node multi-hop bounds, as solving all nodes counting trivially also solves designated node counting.

\subsection{Lower Bounds via Reduction from Contention Resolution}\label{subsec-lower-bound-reduction}

We begin with our lower bound arguments that rely on reductions from contention resolution. For the single-hop scenario, we can reduce from single-hop  contention resolution and apply existing lower bounds from~\cite{newport14}. (Due to space constraint, see Appendix \ref{subsec-appx-omit-lower-contention-resolve-singlehop} for details on contention resolution lower bounds in single-hop networks.) For multi-hop designated node counting, however, we must first prove new contention resolution lower bounds.

In particular, consider the definition of multi-hop contention resolution in which there is a well-defined designated node $w$, and the goal is for exactly one of $w$'s neighbors---which is a size $n_w$ subset drawn from a size $N_\Delta$ universe---to broadcast alone in some time slot. At first glance, this problem might seem easier than single-hop contention resolution as we are provided with a designated node $w$ that could coordinate its neighbors in their quest to break symmetry among themselves. We prove, however, that this is not the case: the lower bounds are the same as their single-hop counterparts. In more detail, we prove the following two lemmas by adapting the techniques from~\cite{newport14} to this new set of assumptions (see Appendix \ref{subsec-appx-omit-lower-proof} for the omitted proofs of this section):

\begin{lemma}\label{lemma-contention-resolution-multihop-no-cd}
Let $\mathcal{A}$ be an algorithm that solves contention resolution in $g(N_{\Delta})$ time slots with probability $p$ in multi-hop networks with no collision detection. It follows that: (a) if $p$ is some constant, then $g(N_{\Delta})\in\Omega(\lg{N_{\Delta}})$; and (b) if $p\geq 1-1/N_{\Delta}$, then $g(N_{\Delta})\in\Omega(\lg^2{N_{\Delta}})$.
\end{lemma}

\begin{lemma}\label{lemma-contention-resolution-multihop-cd}
Let $\mathcal{A}$ be an algorithm that solves contention resolution in $g(N_{\Delta})$ time slots with probability $p$ in multi-hop networks with collision detection. It follows that if $p$ is some constant, then $g(N_{\Delta})\in\Omega(\lg\lg{N_{\Delta}})$.
\end{lemma}

With the needed contention resolution lower bounds in hand, we turn our attention to reducing this problem to approximate neighbor counting. Take the single-hop scenario as an example, the basic idea behind the reduction is that once nodes have an estimate $\hat{n}$ of $n$, they can simply broadcast with probability $1/\hat{n}$ in each time slot. If this estimate is good, then in each time slot, they have a constant probability of isolating a broadcaster, thus solving contention resolution. Moreover, repeating this step multiple times increases the chance of success proportionally. Building on these basic observations, we prove the following:

\begin{lemma}\label{lemma-counting-and-contention-link}
Assume there exists an algorithm $\mathcal{A}$ that solves approximate neighbor counting in $h(N)$ (or, $h(N_\Delta)$ in the multi-hop scenario) time slots with probability $p$. Then, there exists an algorithm $\mathcal{B}$ that solves contention resolution in $2(h(N)+k)$ (resp., $2h(N_\Delta)+k$ in the multi-hop scenario) time slots with probability at least $(1-e^{-k/(4\tilde{c})})\cdot p$. Here, $k\geq 1$ is an integer, and $\tilde{c}\geq 1$ is the constant defined in Section \ref{sec-model}.
\end{lemma}

Combining the reduction described in Lemma \ref{lemma-counting-and-contention-link} with the single-hop lower bounds for contention resolution from~\cite{newport14} and the new multi-hop lower bounds proved above, we get the following lower bounds for approximate neighbor counting:

\begin{theorem}\label{thm-counting-lower-bound-part1}
In a single-hop radio network containing at most $N$ nodes:

\begin{itemize}[itemsep=0.5pt, topsep=0.5pt, parsep=0.5pt]
	\item When collision detection is not available, solving approximate neighbor counting with constant probability requires $\Omega(\lg{N})$ time in the worst case; solving approximate neighbor counting with high probability in $N$ requires $\Omega(\lg^2{N})$ time in the worst case.
	\item When collision detection is available, solving approximate neighbor counting with constant probability requires $\Omega(\lg\lg{N})$ time in the worst case.
\end{itemize}

\noindent In a multi-hop radio network in which the designated node has at most $N_{\Delta}$ neighbors:

\begin{itemize}[itemsep=0.5pt, topsep=0.5pt, parsep=0.5pt]
	\item When collision detection is not available, solving approximate neighbor counting with constant probability requires $\Omega(\lg{N_{\Delta}})$ time in the worst case; solving approximate neighbor counting with high probability in $N_{\Delta}$ requires $\Omega(\lg^2{N_{\Delta}})$ time in the worst case.
	\item When collision detection is available, solving approximate neighbor counting with constant probability requires $\Omega(\lg\lg{N_{\Delta}})$ time in the worst case.
\end{itemize}
\end{theorem}

\subsection{Custom Lower Bounds for High Probability and Collision Detection}\label{subsec-lower-bound-new}

At this point, for single-hop and designated node multi-hop variants of the approximate neighbor counting problem, the only lower bounds missing are the ones of ensuring high success probability with collision detection. As we detail in Appendix \ref{subsec-appx-omit-lower-why-reduc-fail}, our previous reduction-based approach no longer works in these scenarios. (Roughly speaking, the reduction itself takes at least as long as the lower bound we intend to prove.) Therefore, we must construct custom lower bounds for this problem and exact set of assumptions.

We start by proving the following combinatorial result:

\begin{lemma}\label{lemma-comb}
Let $c$ and $k$ be two positive integers such that $c\leq k$. Let $\mathcal{R}$ be the set containing all size $c$ subsets from $[k]=\{1,2,\cdots,k\}$. Let $\mathcal{H}$ be an arbitrary set of size less than $\lg(k/c)$ such that each element in $\mathcal{H}$ is a subset of $[k]$. Then, there exists some $R\in\mathcal{R}$ such that for each $H\in\mathcal{H}$, either $R\subseteq H$ or $R\cap H=\emptyset$.
\end{lemma}

Intuitively, a set $\mathcal{H}$ can be interpreted as a \emph{broadcast schedule} generated by an algorithm $\mathcal{A}$: a node labeled $i$ broadcasts in slot $j$ if and only if it is activated, and $i$ is in the $j$\textsuperscript{th} set in $\mathcal{H}$. Given this interpretation, Lemma \ref{lemma-comb} suggests: for both the single-hop and the multi-hop designated node scenario, for any approximate neighbor counting algorithm $\mathcal{A}$, and for any broadcast schedule generated by $\mathcal{A}$ of length less than $\lg{(k/c)}$, there exists a set of $c$ nodes (or a set of $c$ neighbors of the designated node in the multi-hop scenario) such that if these $c$ nodes are activated and execute $\mathcal{A}$, then during each of the first $\lg{(k/c)}-1$ time slots, either none of them broadcast or all of them broadcast. This further implies, if only two of these $c$ nodes are activated, then their view (of the first $\lg{(k/c)}-1$ time slots of the execution) is indistinguishable from the case in which all of these $c$ nodes are activated.

Now, imagine an adversary who samples a size $c$ subset from $\mathcal{R}$ with uniform randomness, and then flips a fair coin to decide whether to activate all these $c$ nodes, or just two of them. If the adversary happens to have chosen the set $R$ proved to exist in Lemma \ref{lemma-comb}, then by the end of slot $\lg{(k/c)}-1$, algorithm $\mathcal{A}$ cannot distinguish between two and $c$ nodes. Notice, if $c$ is large compared to the approximation threshold $\tilde{c}$, then this difference matters: outputting two when the real count is $c$ (or vice versa) is unacceptable. Thus, in such case, the algorithm gets the right answer with only probability $1/2$---not enough for high success probability.

A complete and rigorous proof for the above intuition is actually quite involved, again see Appendix \ref{subsec-appx-omit-lower-proof} for more details. In the end, we obtain the following lower bounds:

\begin{theorem}\label{thm-counting-lower-bound-part2}
Assume collision detection is available, then:

\begin{itemize}[itemsep=0.5pt, topsep=0.5pt, parsep=0.5pt]
	\item In a single-hop radio network containing at most $N$ nodes, solving approximate neighbor counting with high probability in $N$ requires $\Omega(\lg{N})$ time in the worst case.
	\item In a multi-hop radio network in which the designated node has at most $N_{\Delta}$ neighbors, solving approximate neighbor counting with high probability in $N_{\Delta}$ requires $\Omega(\lg{N_{\Delta}})$ time in the worst case.
\end{itemize}
\end{theorem}

\subsection{All Nodes Multi-Hop Lower Bounds}\label{subsec-lower-bound-carry-over}

If an algorithm can solve multi-hop all nodes approximate neighbor counting, then clearly the same algorithm can be used to solve multi-hop designated node approximate neighbor counting, with same time complexity and success probability. Therefore, the lower bounds we previously proved for the latter variant naturally carries over to the former variant:

\begin{theorem}\label{thm-counting-lower-bound-part3}
In a multi-hop radio network containing at most $N$ nodes:

\begin{itemize}[itemsep=0.5pt, topsep=0.5pt, parsep=0.5pt]
	\item When collision detection is not available, solving approximate neighbor counting with high probability in $N$ requires $\Omega(\lg^2{N})$ time in the worst case.
	\item When collision detection is available, solving approximate neighbor counting with high probability in $N$ requires $\Omega(\lg{N})$ time in the worst case.
\end{itemize}
\end{theorem}

\section{Upper Bounds}

In this section, we describe and analyze several randomized algorithms that solve the approximate neighbor counting problem. We will begin with single-hop all nodes counting. Specifically, four algorithm are presented for this variant, each based on a different approach. Though most of the strategies used are previously known, extensions to design and analysis are often needed. We then introduce three algorithms for multi-hop all nodes counting, including one which is particularly interesting, as it uses a ``double counting'' trick that is not related to contention resolution at all to obtain high success probability. Finally, we briefly discuss solutions for multi-hop designated node counting, as most of these algorithm are simple variations of their counterparts for single-hop all nodes counting.

Due to space constraint, if not otherwise stated, complete proofs for lemma and theorem statements are provided in the appendix. Nonetheless, we will usually discuss the intuitions or high-level strategies for proving them.

\subsection{Single-Hop Networks: No Collision Detection}

Our algorithms often adopt a classical technique inspired by the contention resolution literature: ``\emph{guess and verify}''. In more detail, take a \emph{guess} about the count, and then \emph{verify} its accuracy; if the guess is good enough then we are done, otherwise take another guess and repeat. This simple approach is versatile: depending on how the guesses are made and verified, many variations exist, resulting in efficient algorithms suitable for different settings.

A standard approach to this guessing is to use a geometric sequence with common ratio two, which is usually called (exponential) \emph{decay}~\cite{bar-yehuda87}. This sequence leverages the fact that we only need a constant factor estimate to speed things up. Particularly, if the real count is $n$, then only $O(\lg{n})$ iterations are needed before reaching an accurate estimate.

Once a guess is made, we need to verify its accuracy. To accomplish this, it is sufficient to let each participating node broadcast with a probability proportional to the reciprocal of the guess, and then observe the status of the channel. The intuition is simple: underestimate will result in collision and overestimate will result in silence; and we expect one node to broadcast alone---a distinguishable event---iff the estimate is accurate enough. The algorithms described below adapt this general approach to their specific constraints.

\subparagraph*{Constant probability of success.} We now present \CountSHnoCDConst. (In the algorithm's name, \textsc{SH} means ``single-hop'', \textsc{noCD} means ``no collision detection'', and \textsc{Const} means ``success with constant probability''.) This algorithm applies the most basic form of the ``guess and verify'' strategy. It provides a correct estimate with constant probability in $O(\lg{n})$ time for single-hop radio networks, when collision detection is not available.

\CountSHnoCDConst contains multiple iterations, each of which has two time slots. In the $i$\textsuperscript{th} iteration, nodes assume $n\approx 2^{i}$, and verify whether this estimate is accurate or not. More specifically, in the first time slot within the $i$\textsuperscript{th} iteration, each node will broadcast a \textsf{beacon} message with probability $1/2^i$ and listen otherwise. If a node decides to listen and hears a \textsf{beacon} message in the first time slot, then it will set its estimate to $2^{i+2}$ and terminate after this iteration. That is, if a single node $u$ broadcasts alone in the first time slot of iteration $i$, then all listening nodes---which is all nodes except $u$---will terminate by the end of this iteration, with $2^{i+2}$ being their estimate. Notice, we still need to inform $u$ about this estimate, which is the very purpose of the second time slot within each iteration. More specifically, in the second time slot within the $i$\textsuperscript{th} iteration, for each node $u$, if it has heard a \textsf{beacon} message in the first time slot of this iteration, then it will broadcast a \textsf{stop} message with probability $1/(2^i-1)$. Otherwise, if $u$ has broadcast in the first time slot of this iteration, then it will listen in this second time slot. Moreover, it will terminate with its estimate set to $2^{i+2}$, if it hears a \textsf{stop} message in this second time slot.

Despite its simplicity, proving the correctness of \CountSHnoCDConst requires efforts beyond what would suffice for basic contention resolution. First, by carefully calculating and summing up the failure probabilities, we show no node will terminate during the first $\lg{n}-3$ iterations, with at least constant probability. Then, we show during iteration $i$ where $\lg{n}-2\leq i\leq \lg{n}$, either no node terminates or all nodes terminate, with at least constant probability. Finally, we prove that if all nodes are still active by iteration $\lg{n}$, then all of them will terminate by the end of it, again with at least constant probability. The full analysis can be found in Appendix \ref{subsec-appx-omit-CountSHnoCDConst}. Here we present only the main theorem:

\begin{theorem}\label{thm-CountSHnoCDConst}
The \CountSHnoCDConst approximate neighbor counting algorithm guarantees the following properties with constant probability when executed in a single-hop network with no collision detection: (a) all nodes terminate simultaneously within $O(\lg{n})$ slots; and (b) all nodes obtain the same estimate of $n$, which is in the range $[n,4n]$.
\end{theorem}

\subparagraph*{High probability of success.} Observe that in the aforementioned simplest form of ``guess and verify'', as the estimate increases, the probability that multiple nodes broadcast decreases, and the probability that no node broadcasts increases. A more interesting metric is the probability that a single node broadcasts alone: it first increases, and then decreases; not surprisingly, the peak value is reached when the estimate is the real count. These facts suggest, for each estimate $\hat{n}$, we could repeat the procedure of broadcasting with probability $1/\hat{n}$ multiple times, and use the \emph{fraction} of noisy/silent/clear-message slots to determine the accuracy of the estimate. This method can provide stronger correctness guarantees, but complicates \emph{termination detection} (i.e., when should a node stop), as different nodes may observe different fraction values. For example, if some nodes have already obtained a correct estimate but terminate too early, then remaining nodes might never get correct estimates, since there are fewer nodes remaining.
The situation becomes more challenging when an upper bound of the real count is not available. To resolve this issue, sometimes, we have to carefully craft and embed a ``consensus'' mechanism.

\CountSHnoCDHigh highlights our above discussion. This algorithm contains multiple iterations, each of which has three phases. In the $i$\textsuperscript{th} iteration, nodes assume $n\approx 2^{i}$. The first phase of each iteration $i$---which contains $\Theta(i)$ time slots---is used to verify the accuracy of the current estimate. In particular, in each slot within the first phase, each node will broadcast a \textsf{beacon} message with probability $1/2^i$, and listen otherwise. By the end of the first phase, each node will calculate the fraction of time slots (among all its listening slots in this phase) in which it has heard a \textsf{beacon} message. For each node, if at the end of the first phase of some iteration $j$, this fraction value has reached $1/2e$ for the first time since the start of execution, the node will set $2^j$ as its \emph{private estimate} for $n$. Recall nodes might not obtain private estimates simultaneously, thus they cannot simply terminate and output private estimates as the final estimate. This is the place where the latter two phases come into play. More specifically, in the second phase, nodes that have already obtained private estimates will try to broadcast \textsf{informed} messages to signal other nodes to stop. In fact, hearing an \textsf{informed} message is the only situation in which a node can safely terminate, even if the node has already obtained its private estimate. On the other hand, the third phase is used to deal with the case in which one single ``unlucky'' node successfully broadcasts an \textsf{informed} message during phase two (thus terminate all other nodes), but never gets the chance to successfully receive an \textsf{informed} message (thus cannot terminate along with other nodes). The complete description of \CountSHnoCDHigh is given in Appendix \ref{subsec-appx-omit-CountSHnoCDHigh}.

To prove the correctness of \CountSHnoCDHigh, we need to show: (a) nodes can correctly determine the accuracy of their estimates; and (b) all nodes terminate simultaneously and output identical estimate. Part (b) follows from our careful protocol design, as phase two and three in each iteration act like a mini ``consensus'' protocol, allowing nodes to agree on when to stop. Proving part (a), on the other hand, needs more effort. Recall we use the fraction of clear message slots to determine the accuracy of an estimate, and the expected fraction value should be identical for all nodes. However, due to random chances, the actual fraction value observed by each node might deviate from expectation. If we have an upper bound $N$ of $n$, then by making the first phase to contain $\Theta(\lg{N})$ slots, Chernoff bounds~\cite{mitzenmacher05} will enforce the observed fraction value to be tightly concentrated around its expectation. In our case, $N$ is not available, and we rely on more careful analysis. Specifically, during iterations one to $\lg(n/(a\ln{n}))$ where $a$ is some sufficiently large constant, in each time slot in phase one, at least two nodes will broadcast (since the estimate is too small), thus no node will obtain private estimate in these iterations. Starting from iteration $\lg(n/(a\ln{n}))$, the length of phase one is long enough so that concentration inequalities will ensure the observed fraction value is close to its expectation. Building on these observations, we can eventually conclude the following theorem (the full analysis is in Appendix \ref{subsec-appx-omit-CountSHnoCDHigh}):

\begin{theorem}\label{thm-CountSHnoCDHigh}
The \CountSHnoCDHigh approximate neighbor counting algorithm guarantees the following properties with high probability in $n$ when executed in a single-hop network with no collision detection: (a) all nodes terminate simultaneously within $O(\lg^2{n})$ slots; and (b) all nodes obtain the same estimate of $n$, which is in the range $[n,4n]$.
\end{theorem}

\subsection{Single-Hop Networks: Collision Detection}

Without collision detection, the feedback to a listening node is either silence or a message. Failing to receive a message, therefore, does not hint the nature of the failure: either no node is sending, or multiple nodes are sending. As a result, in the two previous algorithms, when nodes ``guess'' the count, they have to do it in a \emph{linear} manner: start with a small estimate, and double if the guess is incorrect. With collision detection, by contrast, listening nodes can distinguish whether too few (i.e., zero) or too many (i.e., at least two) nodes are broadcasting. As first pointed out back in the 1980's~\cite{willard86}, this extra power enables an exponential improvement over linear searching, since nodes can now perform a {\em binary search}.

\subparagraph*{Constant probability of success.} Here we leverage the aforementioned binary search strategy to return a constant factor estimate of $n$ in $O(\lg\lg{n})$ time, with at least some constant probability. Recall efficient binary search requires a rough upper bound of $n$ as input. To this end, we first introduce an algorithm called \EstUpperSH: it can provide a polynomial upper bound of $n$ within $O(\lg\lg{n})$ time. At a high level, \EstUpperSH is doing a linear ``guess and verify'' search to estimate $\lg{n}$. (Notice, it is \emph{not} estimating $n$.) This strategy, to the best of our knowledge, is first discussed by Willard in the seminal paper \cite{willard86}, and has later been used in other works (see, e.g., \cite{chen13,brandes17}). Due to space constraint, detailed description and analysis of \EstUpperSH are provided in Appendix \ref{subsec-appx-omit-EstUpperSH}.

Once this estimate is obtained, we switch to the main logic of \CountSHCDConst. This algorithm contains multiple iterations, each of which has four time slots. In each iteration $i$, all nodes have a lower bound $a_i$ and an upper bound $b_i$, and will test whether the median $m_i=\lfloor (a_i+b_i)/2\rfloor$ is close to $\lg{n}$ or not.
More specifically, in the first time slot in iteration $i$, each node will broadcast a \textsf{beacon} message with probability $1/2^{m_i}$, and listen otherwise. Listening nodes will use the channel status they observed to adjust $a_i$ (or $b_i$), or terminate and output the final estimate. On the other hand, the other three time slots in each iteration allow nodes that have chosen to broadcast in the first time slot to learn the channel status too, with the help of the nodes that have chosen to listen in the first time slot. (See Appendix \ref{subsec-appx-omit-CountSHCDConst} for complete description of \CountSHCDConst.)

To prove \CountSHCDConst can provide a correct estimate, we demonstrate that during one execution of \CountSHCDConst: (a) whenever $m_i$ is too large or too small, all nodes can correctly detect this and adjust $a_i$ or $b_i$ accordingly; and (b) when $m_i$ is a good estimate, all nodes can correctly detect this as well and stop execution. The full analysis is presented in Appendix \ref{subsec-appx-omit-CountSHCDConst}, here we state only the main theorem:

\begin{theorem}\label{thm-CountSHCDConst}
The \CountSHCDConst approximate neighbor counting algorithm guarantees the following properties when executed in a single-hop network with collision detection: (a) all nodes terminate simultaneously; and (b) with at least constant probability, all nodes obtain the same estimate of $n$ in the range $[n,4n]$ within $O(\lg\lg{n})$ time slots.
\end{theorem}

\subparagraph*{High probability of success.} Our last algorithm for the single-hop scenario is called \CountSHCDHigh. It significantly differs from the other algorithms studied so far in that it does {\em not} use a ``guess and verify'' strategy. Instead, it deploys a {\em random walk} to derive an estimate. The use of random walks for contention resolution was introduced by Nakano and Olariu~\cite{nakano02}, in the context of leader election in radio networks. It was later adopted by Brandes et al.~\cite{brandes17} for solving approximate counting in \textsc{beeping} networks.

Prior to executing \CountSHCDHigh, nodes will first use $O(\lg\lg{n})$ time slots to run \EstUpperSH to obtain a polynomial upper bound of $n$. Call this upper bound $\hat{N}$. All nodes then perform a random walk, the state space of which consists of potential estimates of $n$. More specifically, \CountSHCDHigh contains $\Theta(\lg{\hat{N}})$ iterations, each of which has three time slots. In each iteration, all nodes maintain a current estimate on $n$ which is denoted by $\hat{n}$. (Initially, $\hat{n}$ is set to $\hat{N}$.) In the first slot in a iteration, each node will broadcast a \textsf{beacon} message with probability $1/\hat{n}$, and listen otherwise. If a node hears silence, it will decrease $\hat{n}$ by a factor of four; if a node hears noise, it will increase $\hat{n}$ by a factor of four; and if a node hears a \textsf{beacon} message, it will keep $\hat{n}$ unchanged. Similar to what we have done in \CountSHCDConst, in each iteration, the nodes that have chosen to listen in the first time slot will use the latter two slots to help nodes that have chosen to broadcast in the first time slot to learn the channel status of the first time slot. After these $\Theta(\lg{\hat{N}})$ iterations, all nodes will use $4\tilde{n}$ to be the final estimate of $n$, where $\tilde{n}$ is the most frequent estimate used by the nodes during the $\Theta(\lg{\hat{N}})$ iterations.

The high-level intuition of \CountSHCDHigh is: when the estimate is too large or too small, it will quickly shift towards correct estimates; and when the estimate is correct, it will remain unchanged. Therefore, the most frequent estimate will likely to be a correct one. The full analysis is deferred to Appendix \ref{subsec-appx-omit-CountSHCDHigh}, here we provide only the main theorem:

\begin{theorem}\label{thm-CountSHCDHigh}
The \CountSHCDHigh approximate neighbor counting algorithm guarantees the following properties when executed in a single-hop network with collision detection: (a) all nodes terminate simultaneously; and (b) with high probability in $n$, all nodes obtain the same estimate of $n$ in the range $[n,64n]$ within $O(\lg{n})$ time slots.
\end{theorem}

\subsection{Multi-Hop with All Nodes Counting: No Collision Detection}

All nodes counting in a multi-hop network is challenging as different nodes may have significantly different number of neighbors. In this part, we present two algorithms that attempt to overcome this obstacle, the second of which is particularly interesting.

We begin with the first algorithm---called \CountAllnoCDa---which still relies on the linear ``guess and verify'' approach. However, it requires the upper bound $N_{\Delta}$ as an input parameter to enforce termination. Nonetheless, for each node, \CountAllnoCDa always returns an accurate estimate, even when knowledge of $N_{\Delta}$ is absent.

In more detail, \CountAllnoCDa contains $\lg{N_\Delta}$ iterations, and the $i$\textsuperscript{th} iteration contains $\Theta(i)$ time slots. In each slot in iteration $i$, each node will choose to be a broadcaster or a listener each with probability $1/2$. If a node chooses to be a listener in a time slot, it will simply listen. Otherwise, if a node chooses to be a broadcaster, it will broadcast a \textsf{beacon} message with probability $1/2^i$, and do nothing otherwise. After an iteration $i$, for a node $u$, if for the first time since the beginning of protocol execution, it has heard \textsf{beacon} messages in at least $1/10$ fraction of slots among the listening slots (within this iteration), then $u$ will use $2^{i+3}$ as its estimate for $n_u$. Proving the correctness of \CountAllnoCDa borrows heavily from our analysis of \CountSHnoCDHigh (see Appendix \ref{subsec-appx-omit-CountAllnoCDa} for more details), here we only state the main theorem:

\begin{theorem}\label{thm-CountAllnoCDa}
For each node $u$, the \CountAllnoCDa approximate neighbor counting algorithm guarantees the following with high probability in $n_u$ when executed in a multi-hop network with no collision detection: $u$ will obtain an estimate of $n_u$ in the range $[n_u,4n_u]$ within $O(\lg^2{n_u})$ time. Moreover, $u$ will terminate after $O(\lg^2{N_\Delta})$ time when $N_\Delta$ is known.
\end{theorem}

Notice that in \CountAllnoCDa, for a node $u$, the high correctness guarantee is with respect to $n_u$. This implies, when $n_u$ is some constant, the probability that the obtained estimate is desirable is also a constant. Sometimes, we may want \emph{identical} and \emph{high} correctness guarantees for \emph{all} estimates, such as high probability in $n$. Our second algorithm---which is called \CountAllnoCDb---achieves this goal, at the cost of accessing $N$ and demanding longer execution time. (\CountAllnoCDa only needs $N_\Delta$ to enforce termination, while $\CountAllnoCDb$ needs $N$ to work properly.)

Careful readers might suspect \CountAllnoCDb just extends the length of each iteration of \CountAllnoCDa to $\Theta(\lg{N})$. Unfortunately, this simple modification is not sufficient: we still cannot change the fact that when a node $u$ listens, the number of broadcasters among its neighbors is concentrated to $n_u/2$ only with high probability in $n_u$.

Instead, \CountAllnoCDb takes a different approach, the core of which is a ``double counting'' trick. Specifically, \CountAllnoCDb contains $L=\Theta(\lg{N})$ iterations, each of which has $\Theta(\lg^2{N})$ time slots. At the beginning of each iteration, each node chooses to be a broadcaster or a listener each with probability $1/2$. Then, by applying the ``guess and verify'' strategy, each listener will spend the $\Theta(\lg^2{N})$ time slots in this iteration to obtain a constant factor estimate on the number of neighboring broadcasters. When all $\Theta(\lg{N})$ iterations are done, each node will sum the estimates it has obtained, divide it by $L/4$, and output the result as its estimate for the neighborhood size.

Due to space constraint, detailed description for each iteration is deferred to Appendix \ref{subsec-appx-omit-CountAllnoCDb}. We only note here that the estimates obtained by the listeners are quite accurate:

\begin{lemma}\label{lemma-CountAllnoCDb-1}
Consider an arbitrary iteration during the execution of \CountAllnoCDb, assume node $u$ is a listener with $m$ neighboring broadcasters. By the end of this iteration, node $u$ will obtain an estimate of $m$ in the range $[m,4m]$, with high probability in $N$.
\end{lemma}

We can now state and prove the guarantees provided by \CountAllnoCDb:

\begin{theorem}\label{thm-CountAllnoCDb}
The \CountAllnoCDb approximate neighbor counting algorithm guarantees the following properties when executed in a multi-hop network with no collision detection: (a) all nodes terminate after $O(\lg^3{N})$ time slots; and (b) with high probability in $N$, for each node $u$, the node will obtain an estimate of $n_u$ in the range $[n_u,5n_u]$.
\end{theorem}

\begin{proof}[Proof sketch.]
Consider a node $u$ and one of its neighbor $v$. Assume \CountAllnoCDb contains $a\lg{N}$ iterations, where $a$ is a sufficiently large constant. In expectation, in $(a/4)\cdot\lg{N}$ iterations, $u$ will be listener and $v$ will be broadcaster. Apply a Chernoff bound, we know $u$ will be listener and $v$ will be broadcaster in at least $(1-\delta)\cdot(a/4)\cdot\lg{N}$ iterations, and at most $(1+\delta)\cdot(a/4)\cdot\lg{N}$ iteration, w.h.p.\ in $N$. Here, $0<\delta<1$ is a small constant determined by $a$. Take a union bound over all $O(N)$ neighbors of $u$, we know this claim holds true for them as well. Therefore, if $u$ were able to accurately count the number of broadcasting neighbors without any error in each listening iteration, the sum it will obtain would be in the range $[(1-\delta)\cdot(a/4)\cdot\lg{N}\cdot n_u, (1+\delta)\cdot(a/4)\cdot\lg{N}\cdot n_u]$, w.h.p.\ in $N$.

Now, due to Lemma \ref{lemma-CountAllnoCDb-1}, we know the actual sum of counts $u$ will obtain is in the range $[(1-\delta)\cdot(a/4)\cdot\lg{N}\cdot n_u, 4\cdot(1+\delta)\cdot(a/4)\cdot\lg{N}\cdot n_u]$, w.h.p.\ in $N$. As a result, according to our algorithm description, the final estimate $u$ will obtain is in the range $[n_u, 5n_u]$, w.h.p.\ in $N$. Take a union bound over all nodes, the theorem is proved.
\end{proof}

\subsection{Multi-Hop with All Nodes Counting: Collision Detection}

In \CountAllnoCDa, we resolve the termination detection problem by accessing $N_\Delta$. This allows nodes to run long enough so that they could be sure that everyone has a chance to learn what it needed to learn. With the addition of collision detection, however, the assumption that nodes know $N_{\Delta}$ can be removed for many network topologies. In particular, we can leverage the idea that neighbors of $u$ that {\em have not} obtained estimates yet can use noise to \emph{reliably} inform $u$ that they wish $u$ to continue. We call this algorithm \CountAllCDa.

\CountAllCDa contains multiple iterations, each of which has two parts. The first part of any iteration $i$ is identical to iteration $i$ of \CountAllnoCDa. The second part, on the other hand, helps nodes to determine when to stop. In particular, the second part of iteration $i$ contains a single slot. For a node $u$, if it has not obtained an estimate of $n_u$ by the end of the first part of iteration $i$ yet, then in the second part, it will broadcast a \textsf{continue} message. On the other hand, if $u$ has already obtained an estimate of $n_u$ by the end of the first part of iteration $i$, it will simply listen in part two. Moreover, $u$ will continue into the next iteration iff it hears \textsf{continue} or noise during part two. The guarantees provided by \CountAllCDa are stated below, and the proof of it is provided in Appendix~\ref{subsec-appx-omit-CountAllCDa}.

\begin{theorem}\label{thm-CountAllCDa}
The \CountAllnoCDa approximate neighbor counting algorithm guarantees the following properties for each node $u$ when executed in a multi-hop network with collision detection: (a) $u$ will obtain an estimate of $n_u$ in the range $[n_u,4n_u]$ within $O(\lg^2{n_u})$ time slots, with high probability in $n_u$; and (b) if $\sum_{v\in\Gamma_u\cup\{u\}}(1/{n_v}) < 1$, then $u$ will terminate within $(\max_{v\in\Gamma_u\cup\{u\}}\{\lg{n_v}\})^2$ time slots, with probability at least $1-\sum_{v\in\Gamma_u\cup\{u\}}(1/{n_v})$.
\end{theorem}

A key point about the above termination bound is that it requires $\sum_{v\in\Gamma_u\cup\{u\}}(1/{n_v}) < 1$. (In fact, this constraint can be relaxed to $\sum_{v\in\Gamma_u\cup\{u\}}(1/{n^\gamma_v}) < 1$, for an arbitrarily chosen constant $\gamma\geq 1$.) If this is not the case (e.g., in a dense star network), the termination detection mechanism might not work properly. In that situation, the default dependence on $N_{\Delta}$ from the no collision detection case can be applied as a back-up.

\subsection{Multi-Hop with Designated Node Counting}

Compared to the all nodes counting variant, multi-hop neighbor counting with only the designated node is easier: the strategies we previously used for single-hop counting are still applicable, and the introduction of the designated node can actually make coordination easier. (In particular, this node can greatly simplify termination detection). Due to space constraint, we defer the upper bounds for this variant to Appendix \ref{sec-appx-omit-alg-center}.

We note that one interesting algorithm in this variant is \CountCenternoCDConst, which achieves constant success probability without collision detection. \CountCenternoCDConst differs from its single-hop counterpart (i.e., \CountSHnoCDConst) in that it uses fraction of clear message slots to determine the accuracy of nodes' estimate. The primary reason we develop this algorithm is that the success probability of \CountSHnoCDConst is \emph{fixed}. In contrast, in \CountCenternoCDConst, by tweaking the running time (up to some constant factor), the success probability---despite being a constant---can be adjusted accordingly. More details of this algorithm are presented in Appendix \ref{subsec-appx-omit-alg-center-noCD}.



\section{Discussion}

We see at least two problems that worth further exploration. First, how do termination requirements affect the complexity of the problem? This is particularly interesting in the multi-hop all nodes counting scenario, when knowledge of $N_\Delta$ or $N$ is not available. \CountAllnoCDa shows lower bound can be achieved at the cost of no termination, but how much time must we spend if termination needs to be enforced, or is simply impossible without knowing $N_\Delta$ or $N$? Another open problem concerns the gap between the lower and upper bounds, in the multi-hop all nodes counting scenario with collision detection. On the one hand, the lower bound might be loose as it is a simple carry over, ignoring the possibility that all nodes counting could be fundamentally harder than designated node counting. Yet on the other hand, we have not found a way to leverage collision detection to reduce algorithm runtime (e.g., it seems hard to run multiple instances of binary search in parallel). Currently, our best guess is that \emph{both} the lower bound and the upper bound are not tight.

\bibliographystyle{plainurl}
\bibliography{arxiv-full}

\clearpage
\appendix
\section*{Appendix}

\section{Omitted Description and Analysis of Lower Bound Results}\label{sec-appx-omit-lower}

\subsection{The $k$-hitting game and lower bounds for contention resolution in single-hop networks}\label{subsec-appx-omit-lower-contention-resolve-singlehop}

Newport~\cite{newport14} introduced the simple combinatorial {\em $k$-hitting game}, which acts as a flexible and generic ``wrapper'' for the notion of ``hitting sets'' proposed by Alon et al.\ in their seminal paper on centralized broadcast lower bounds \cite{alon91}. Through carefully-crafted reduction arguments, one can reduce $k$-hitting game to many different variations of contention resolution---allowing a fixed lower bound on hitting to carry over to many different contention resolution assumptions. Because we utilize and modify these existing bounds to study neighbor counting, we review the relevant definitions and results here.

In the $k$-hitting game where $k>1$ is an integer, there is one player and one referee. Before the game starts, the referee privately selects a \emph{target set} $T\subseteq\{1,2,\cdots,k\}$. The game then proceeds in rounds. In each round, the player submits a \emph{proposal} $P\subseteq\{1,2,\cdots,k\}$ to the referee. If $|P\cap T|=1$, the player wins. Otherwise, the referee simply tells the player that $P$ is incorrect, and the game proceeds into the next round.

Intuitively, in a radio network, $T$ denotes the $n$ nodes that are activated from ``the universe of $k$ possible nodes''. The set of nodes that would broadcast in a given round if all nodes were running the given algorithm implicitly defines a proposal $P$. In this case, $|P \cap T| = 1$ is equivalent to isolating a broadcaster among $T$, solving contention resolution.

In \cite{newport14}, the author proves that winning the $k$-hitting game with constant probability requires at least $\Omega(\lg{k})$ rounds, while solving it with probability at least $1-1/k$ requires at least $\Omega(\lg^2{k})$ rounds.

\begin{lemma}[\cite{newport14}]\label{lemma-k-hitting}
Fix some player $\mathcal{P}$ that guarantees, for all $k>1$, to win the $k$-hitting game in $f(k)$ rounds with probability at least $p$. It follows that: (a) if $p$ is some constant, then $f(k)\in\Omega(\lg{k})$; and (b) if $p\geq 1-1/k$, then $f(k)\in\Omega(\lg^2{k})$.
\end{lemma}

Then, by reduction arguments, the author is able to extend these lower bounds to handle contention resolution in a single-hop radio network with and without collision detection.

\begin{lemma}[\cite{newport14}]\label{lemma-contention-resolution-singlehop-no-cd}
Let $\mathcal{A}$ be an algorithm that solves the contention resolution problem in $g(N)$ time slots with probability $p$ in the single-hop network model with no collision detection. It follows that: (a) if $p$ is some constant, then $g(N)\in\Omega(\lg{N})$; and (b) if $p\geq 1-1/N$, then $g(N)\in\Omega(\lg^2{N})$.
\end{lemma}

\begin{lemma}[\cite{newport14}]\label{lemma-contention-resolution-singlehop-cd}
Let $\mathcal{A}$ be an algorithm that solves the contention resolution problem in $g(N)$ time slots with probability $p$ in the single-hop network model with collision detection. It follows that: if $p$ is some constant, then $g(N)\in\Omega(\lg\lg{N})$.
\end{lemma}

\subsection{Discussions on why the reduction approach fail in certain cases}\label{subsec-appx-omit-lower-why-reduc-fail}

Recall the two remaining cases are: all nodes counting in single-hop radio networks and designated node counting in multi-hop radio networks, when collision detection is available and success with high probability is required. Careful readers might have already realized that the reduction approach can still be applied here. Unfortunately, however, this approach would no longer give the desired results.

More specifically, an $\Omega(\lg{N})$ lower bound for contention resolution with high success probability in single-hop scenario with collision detection is already given in \cite{newport14}, and we believe an $\Omega(\lg{N_\Delta})$ lower bound for contention resolution (of the designated node's neighbors) with high probability in multi-hop scenario with collision detection could also be derived. Moreover, with Lemma \ref{lemma-counting-and-contention-link}, in these two cases, we can still link the complexities of contention resolution and approximate neighbor counting together. However, the critical issue is that given an efficient algorithm for approximate neighbor counting, the overhead incurred by utilizing this algorithm to solve contention resolution is too high. In particular, the lower bounds we intend to prove (for approximate neighbor counting) are $\Omega(\lg{N})$ and $\Omega(\lg{N_\Delta})$, respectively; yet the overhead incurred by the reduction process have already reached $O(\lg{N})$ and $O(\lg{N_\Delta})$ (see Lemma \ref{lemma-counting-and-contention-link}), respectively.

\subsection{Omitted proofs}\label{subsec-appx-omit-lower-proof}

\begin{proof}[Proof of Lemma \ref{lemma-contention-resolution-multihop-no-cd}.]
We claim if there exists an algorithm $\mathcal{A}$ that solves the contention resolution problem in $g(N_{\Delta})$ time slots with probability $p$, then there exists an algorithm $\mathcal{B}$ that allows a player $\mathcal{P}$ to win the $N_{\Delta}$-hitting game in $g(N_{\Delta})$ rounds with probability at least $p$. By Lemma \ref{lemma-k-hitting}, this claim implies our lemma.

We now prove the above claim. Assume the target set chosen by the referee is $T\subseteq\{1,2,\cdots,N_{\Delta}\}$, further assume $T$ is of size $l$ and contains elements $t_1,t_2,\cdots,t_l$. Imagine a star network $G_*$ in which the designated node is $w$, and $w$ has $l$ neighbors whose identities are $t_1,t_2,\cdots,t_l$. (These identities are generated and used by the player, and nodes in the network do not have access to them.) Now, consider the following strategy. The player simulates running $\mathcal{A}$ in $G_*$. In each round, if any neighbors of $w$ decides to broadcast, then the player generates the proposal according to the identities of these nodes. If the proposal is correct, then we are done. If the proposal is incorrect and $w$ is listening, the player simulates $w$ hearing nothing. If no neighbor of $w$ broadcasts and $w$ is listening, then the player again simulates $w$ hearing nothing. On the other hand, if in a round $w$ broadcasts, then the player simulates listening neighbors hear the message. Otherwise, if in a round $w$ remains silent, then the player simulates listening neighbors hear nothing. Clearly, this simulation correctly reflects how $\mathcal{A}$ would proceed if $G_*$ is a real radio network.

Now, consider the proposals generated by the player. If in a round $\mathcal{A}$ solves the contention resolution problem, then the proposal $P$ in that round must be of size one, which implies $|P\cap T|=1$. That is, this proposal will also let the player win the $N_{\Delta}$-hitting game.

By now, we know if $\mathcal{P}$ can construct $G_*$ and knows an algorithm $\mathcal{A}$ that solves the contention resolution problem in $g(N_{\Delta})$ time slots with probability $p$, then $\mathcal{P}$ also has a strategy to win the $N_{\Delta}$-hitting game in $g(N_{\Delta})$ rounds with probability at least $p$.

However, a critical issue is that the player $\mathcal{P}$ cannot construct $G_*$ directly! In particular, he does not know the size of the target set; he also does not know $t_1,t_2,\cdots,t_l$. In fact, once he knows $t_1,t_2,\cdots,t_l$, he can simply propose $t_1$ and wins the game in a single round.

To overcome this difficulty, $\mathcal{P}$ will simulate running $\mathcal{A}$ on a different star network $G'_*$---one that he can construct directly. In $G'_*$, the designated node $w$ has $N_{\Delta}$ neighbors named $1$ to $N_{\Delta}$. In each round, if any neighbor of $w$ decides to broadcast, $\mathcal{P}$ generates the proposal according to the identities of these nodes. Moreover, for each round, the simulation rules are the same with the ones described above for $G_*$.

Now, the crucial observation is, until the proposal $\mathcal{P}$ submits contains exactly one element in $T$ (i.e., by which point $\mathcal{P}$ wins), for each node in $G_*$, the two execution histories it sees in $G_*$ and $G'_*$ are identical, assuming the node uses the same random bits in each of these two executions. (An interesting point worth noting is, the simulation of $G'_*$ might be \emph{inconsistent} with how $\mathcal{A}$ would proceed in $G'_*$ in real. For example, when a node not in $T$ broadcasts alone in $G'_*$, we still simulate $w$ as receiving nothing. Nonetheless, such inconsistency is fine, so long as we ensure for each node in $G_*$, its views in $G_*$ and $G'_*$ are identical.)

To prove the above crucial observation, we do an induction on the simulated time slots. Prior to the first simulated slot, the claim trivially holds. Assume by the end of slot $s\geq 0$ the claim still holds, we now consider slot $s+1$. First, focus on the designated node $w$. According to the induction hypothesis, $w$ has identical execution histories in $G_*$ and $G'_*$, till the end of slot $s$. Thus, in slot $s+1$, $w$ will perform same action (i.e., broadcast or listen). Particularly, if $w$ broadcasts, then it sends identical messages in $G_*$ and $G'_*$. On the other hand, if $w$ listens, since $s+1$ is a slot prior to $\mathcal{P}$ winning the game, we know either no neighbor of $w$ in $G_*$ broadcasts or at least two neighbors of $w$ in $G_*$ broadcast. In both cases $w$ hears nothing, in both $G_*$ and $G'_*$. Next, consider an arbitrary neighbor $u$ of $w$ that is in $G_*$. According the induction hypothesis, $u$ has identical execution histories in $G_*$ and $G'_*$, till the end of slot $s$. Thus, in slot $s+1$, $u$ will perform same action. Particularly, if $u$ broadcasts, then it sends identical messages in $G_*$ and $G'_*$. On the other hand, if $u$ listens, then the situation depends on whether the designated node $w$ broadcasts in slot $s+1$. In case $w$ remains silent, $u$ hears silence in both $G_*$ and $G'_*$. If $w$ broadcasts, then according to our previous analysis, $w$ will send identical messages in $G_*$ and $G'_*$. This implies $u$ will receive identical messages in $G_*$ and $G'_*$. The proof for the inductive step thus completes.

Since the execution histories in $G_*$ and $G'_*$ are identical for nodes in $G_*$, and since $\mathcal{A}$ solves the contention resolution problem in $g(N_{\Delta})$ time slots with probability $p$ in $G_*$, we know by the end of round $g(N_{\Delta})$, with probability at least $p$, the player must have generated a proposal $P=P_1\cup P_2$, where $P_1\subseteq T$ and $|P_1\cap T|=1$, and $P_2\subseteq\{1,2,\cdots,N_{\Delta}\}\backslash T$. This proposal will let $\mathcal{P}$ win the $N_{\Delta}$-hitting game.
\end{proof}

\begin{proof}[Proof sketch of Lemma \ref{lemma-contention-resolution-multihop-cd}.]
We claim if there exists an algorithm $\mathcal{A}$ that solves the contention resolution problem with collision detection in $g(N_{\Delta})$ time slots with probability $p$, then there exists an algorithm $\mathcal{B}$ that allows a player $\mathcal{P}$ to win the $N_{\Delta}$-hitting game in $2^{g(N_{\Delta})}-1$ rounds with probability at least $p$. Together with Lemma \ref{lemma-k-hitting}, and the fact that $2^{g(N_\Delta)}-1\in o(\lg{N_\Delta})$ if $g(N_\Delta)\in o(\lg\lg{N_\Delta})$, our lemma is immediate. The reminder of this proof is dedicated to proving this claim.

Assume the target set chosen by the referee is $T=\{t_1,t_2,\cdots,t_l\}\subseteq\{1,2,\cdots,N_{\Delta}\}$, where $1\leq l\leq N_{\Delta}$. Imagine the following two star networks: the first one is called $G_*$ in which the designated node $w$ has $l$ neighbors with identities $t_1,t_2,\cdots,t_l$; and the second one is called $G'_*$ in which the designated node $w$ has $N_{\Delta}$ neighbors with identities $1,2,\cdots,N_{\Delta}$.

To win the $N_{\Delta}$-hitting game, the player $\mathcal{P}$ first builds a complete binary tree $B_{g(N_{\Delta})}$ of depth $g(N_{\Delta})-1$, where the root is at level zero. For each non-root node in the tree, we attach a binary label to it according to the following rule: if the node is the left child of its parent then it has label zero, otherwise it has label one. For the ease of presentation, for a node $v$ in the tree, we use $d(v)$ to denote its depth, and $p(v)$ to denote the binary string generated by concatenating the labels of the nodes on the path from root to $v$. For the root node, this binary string is simply an empty string.

Now, for each node $v$ in $B_{g(N_{\Delta})}$, the player $\mathcal{P}$ simulates running algorithm $\mathcal{A}$ in $G'_*$ for $d(v)+1$ time slots, with same sequence of random bits (for each node in $G'_*$), according to the following rules. In the $i$\textsuperscript{th} slot where $1\leq i\leq d(v)$, if the $i$\textsuperscript{th} bit in $p(v)$ is zero, then $\mathcal{P}$ simulates $w$ in $G'_*$ hearing silence in case it is listening; and if the $i$\textsuperscript{th} bit in $p(v)$ is one, then $\mathcal{P}$ simulates $w$ in $G'_*$ hearing collision in case it is listening. One the other hand, for each listening neighbor in $G'_*$, $\mathcal{P}$ simulates it hears whatever $w$ broadcasts, or silence in case $w$ does not broadcast anything. In the $(d(v)+1)$\textsuperscript{st} time slot, $\mathcal{P}$ proposes the identities of neighbors of $w$ in $G'_*$ that decide to broadcast to the referee.

Recall that if we run $\mathcal{A}$ in a real radio network with topology identical to $G_*$, we can solve contention resolution in $g(N_{\Delta})$ time slots with probability $p$. Since the network is fixed, the only uncertainty comes from the random choices made by $\mathcal{A}$ during the execution. Assume the sequence of random bits used by $\mathcal{A}$ is $\pi$. If this execution indeed solves contention resolution within $g(N_{\Delta})$ time slots, then without loss of generality, assume it solves contention resolution at slot $g'(N_{\Delta})\leq g(N_{\Delta})$. This implies, for each time slot prior to $g'(N_{\Delta})$, the designated node $w$ hears either silence or noise, if it is listening. As a result, this further implies, there must exist a node $v$ in tree $B_{g(N_{\Delta})}$, for slots up to (and including) $g'(N_{\Delta})-1$, for nodes in $G_*$, the simulation of $\mathcal{A}$ (in $G'_*$) according to $p(v)$ is consistent with running $\mathcal{A}$ in $G_*$ in real (when using $\pi$ as the source of randomness). (A more rigorous proof for this claim can be obtained via induction on simulated time slots, which is very similar to the one we have done in the proof for Lemma \ref{lemma-contention-resolution-multihop-no-cd}.) Hence, when processing $v$ in tree $B_{g(N_{\Delta})}$, the proposal generated by $\mathcal{P}$ will win the $N_{\Delta}$-hitting game.

(In essence, to correctly simulate running $\mathcal{A}$ in $G_*$ while the actual topology is $G'_*$, the hard scenario is when multiple neighbors of $w$ broadcast yet the resulting proposal does not win the game, since the simulator cannot distinguish between the case where none of the broadcasters were in $T$ (in which case, the correct thing is to simulate $w$ hearing silence) and the case where multiple broadcasters were in $T$ (in which case, the correct thing is to simulate $w$ hearing noise). The binary tree $B_{g(N_{\Delta})}$ enables the simulator to essentially guess at the sequence of $w$'s collision detection information. One of these guesses must be right for the particular definition of $T$ and random bits used in the execution.)

Since a complete binary tree of depth $g(N_{\Delta})-1$ contains $2^{g(N_{\Delta})}-1$ nodes, we know $\mathcal{P}$ can win the $N_{\Delta}$-hitting game with probability at least $p$ in $2^{g(N_{\Delta})}-1$ rounds. This completes the proof of the lemma.
\end{proof}

\begin{proof}[Proof sketch of Lemma \ref{lemma-counting-and-contention-link}.]
We first construct $\mathcal{B}$ in the single-hop scenario.

In this situation, $\mathcal{B}$ contains multiple steps, each of which has two time slots. In odd slots, all nodes simply run $\mathcal{A}$. In even slots, if a node has already obtained the estimate $\hat{n}$ of $n$, it will broadcast with probability $1/\hat{n}$; otherwise, it will remain silent. Since $\mathcal{A}$ solves approximate neighbor counting in $h(N)$ time with probability $p$, we know starting from step $h(N)+1$, with probability $p$, in each even slot, each node will broadcast with probability $1/\hat{n}$, where $n\leq\hat{n}\leq\tilde{c}\cdot n$.

Assume starting from step $h(N)+1$, in each even slot, each node indeed broadcasts with probability $1/\hat{n}$. Thus, in the second slot in each such step, the probability that some node will broadcast alone is $n\cdot(1/\hat{n})\cdot(1-1/\hat{n})^{n-1}\geq (1/\tilde{c})\cdot 4^{-n/\hat{n}}\geq (1/\tilde{c})\cdot 4^{-1}=1/(4\tilde{c})$. That is, in each such slot, the probability that some node will broadcast alone (thus solving the contention resolution problem) is a constant. Since each such slot is independent, our lemma follows in the single-hop scenario.

Next, we turn our attention to the multi-hop scenario.

In this situation, $\mathcal{B}$ contains two parts. In the first part, there are multiple steps, each of which has two slots. In odd slots, all nodes simply run $\mathcal{A}$; and in even slots, all non-designated nodes will listen while $w$ does nothing. According to the assumption, after $h(N_\Delta)$ steps, with probability $p$, the designated node $w$ will obtain a constant factor estimate $\hat{n}_w$ of $n_w$ such that $n_w\leq\hat{n}_w\leq\tilde{c}\cdot n_w$. Once $w$ obtains $\hat{n}_w$, in the next even slot, it will broadcast this estimate. Clearly, $\hat{n}_w$ will be successfully received by all neighbors of $w$. This marks the end of part one of $\mathcal{B}$. The second part of $\mathcal{B}$ is simple: in each slot, each non-designated node which knows $\hat{n}_w$ broadcasts with probability $1/\hat{n}_u$, and $w$ simply listens.

By a similar analysis as in the single-hop scenario, we know in each slot in part two of $\mathcal{B}$, the probability that some neighbor of $w$ broadcasts alone (thus solving the contention resolution problem) is at least $1/(4\tilde{c})$. Since each slot is independent, our lemma follows in the multi-hop scenario as well.
\end{proof}

\begin{proof}[Proof sketch of Theorem \ref{thm-counting-lower-bound-part1}.]
As an example, we consider the single-hop without collision detection scenario. The proofs for other scenarios are very similar.

First, consider the success with constant probability case.

For the sake of contradiction, assume there exists an algorithm $\mathcal{A}$ that solves approximate neighbor counting in $h(N)\in o(\lg{N})$ time in this scenario, with constant probability $p$. Due to Lemma \ref{lemma-counting-and-contention-link}, this implies we can devise an algorithm $\mathcal{B}$ that solves contention resolution in $2(h(N)+k)$ time, with probability at least $(1-e^{-k/(4\tilde{c})})\cdot p$. So long as $k\geq 1$ is some constant, we know $(1-e^{-k/(4\tilde{c})})\cdot p$ will be a constant, and $2(h(N)+k)\in o(\lg{N})$. However, this contradicts the lower bounds shown in Lemma \ref{lemma-contention-resolution-singlehop-no-cd}.

We next consider the success with high probability case.

For the sake of contradiction, assume there exists an algorithm $\mathcal{A}$ that solves approximate neighbor counting in $h(N)\in o(\lg^2{N})$ time in this scenario, with high probability in $N$. I.e., with a probability $p\geq 1-1/N^{\epsilon}$ for some constant $\epsilon\geq 1$. Due to Lemma \ref{lemma-counting-and-contention-link}, this implies we can devise an algorithm $\mathcal{B}$ that solves contention resolution in $2(h(N)+k)$ time, with probability at least $(1-e^{-k/(4\tilde{c})})\cdot p$. By setting $k=4\tilde{c}\cdot\ln{N^\epsilon}$, we know $\mathcal{B}$ solves contention resolution in $o(\lg^2{N})$ time, with probability at least $1-2/N^\epsilon$. Due to proof of Theorem 4 shown in \cite{newport14}, by setting $N'=N/2$, we know $\mathcal{B}$ can be used to devise an algorithm that wins $N'$-hitting game in $o(\lg^2{N'})$ time, with probability at least $1-1/(N')^\epsilon$. I.e., with high probability in $N'$. However, this contradicts the lower bounds shown in Lemma \ref{lemma-k-hitting}.
\end{proof}

\begin{proof}[Proof of Lemma \ref{lemma-comb}.]
Order the sets in $\mathcal{H}$ in some arbitrary manner. For every $i\in [k]$, we associate a binary string $s_i$ of length $|\mathcal{H}|$ in the following manner: the $j$\textsuperscript{th} bit of $s_i$ is one iff $i$ is in the $j$\textsuperscript{th} set of $\mathcal{H}$. Since $|\mathcal{H}|<\lg{(k/c)}$ we know there are less than $k/c$ distinct binary strings of length $|\mathcal{H}|$, and since we use these binary string to label $|[k]|=k$ distinct values, we know there must exist some binary string that is associated with at least $c$ distinct values in $[k]$. Let $R$ be a set containing $c$ of these values, we know $R\in\mathcal{R}$. Since the values in $R$ are associated with the same binary string, by the definition of this binary string, we know for each $H\in\mathcal{H}$, either $R\subseteq H$ (in which case the corresponding bit in the binary string is one) or $R\cap H=\emptyset$ (in which case the corresponding bit in the binary string is zero).
\end{proof}

\begin{proof}[Proof of Theorem \ref{thm-counting-lower-bound-part2}.]
We first focus on the single-hop scenario.

Let $\mathcal{A}$ be an arbitrary (and potentially randomized) distributed algorithm for approximate neighbor counting. Assume $G_\alpha$ is a single-hop radio network in which $k\leq N$ nodes are activated and each node has a unique identity in $[k]$. (These identities are for the ease of presentation, the nodes themselves cannot access these identities.) Simulate each of the first $\lg{(k/c)}-1$ time slots according to the following rules: if a node chooses to broadcast, then we simulate it broadcasting the content specified by $\mathcal{A}$; if a node chooses to listen, then we simulate it hearing silence. We call this execution segment of $\mathcal{A}$ as $\alpha$. Let $H_\alpha^i$ denote the identities of the nodes that broadcast in slot $i$, define set $\mathcal{H}_\alpha=\{H_\alpha^1,H_\alpha^2,\cdots,H_\alpha^{\lg{(k/c)}-1}\}$.

According to Lemma \ref{lemma-comb}, there exists a size $c$ subset $R$ of $[k]$ such that for each $H_\alpha^i\in\mathcal{H}_\alpha$, either $R\subseteq H_\alpha^i$ or $R\cap H_\alpha^i=\emptyset$. Assume $G_\beta$ is a single-hop radio network in which the nodes with identities in $R$ are activated. Run $\mathcal{A}$ in $G_\beta$ for $\lg{(k/c)}-1$ time slots, call the resulting execution segment $\beta$. Notice, if $\mathcal{A}$ is randomized, then for each node, assume the random bits generated for that node are identical in $\alpha$ and $\beta$.

The third single-hop radio network $G_\gamma$ contains two arbitrary nodes with identities in $R$. We run $\mathcal{A}$ in $G_\gamma$ for $\lg{(k/c)}-1$ time slots and call the resulting execution segment $\gamma$. Again, if $\mathcal{A}$ is randomized, then for each node, assume the random bits generated for that node are identical in $\alpha$, $\beta$, and $\gamma$.

Now, the critical claim is, for each node in $G_\beta$, its views in $\alpha$ and $\beta$ are identical. I.e., for each node with an identity in $R$, execution segments $\alpha$ and $\beta$ are indistinguishable. We prove this by a slot to slot induction.

In the first time slot, if a node $u_j$ with identity $j$ chooses to broadcast in $\beta$, then clearly it will also choose to broadcast in $\alpha$. Moreover, $u_j$ will broadcast same content in $\alpha$ and $\beta$. On the other hand, in the first time slot, if $u_j$ chooses to listen in $\beta$, then $u_j$ will also choose to listen in $\alpha$. This means $R\nsubseteq H_\alpha^1$, implying $R\cap H_\alpha^1=\emptyset$. That is, all nodes with identities in $R$ will choose to listen in the first time slot in $\alpha$, which in turn means all nodes in $G_\beta$ will choose to listen in the first time slot. Thus, $u_j$ will hear silence in the first time slot in $\beta$; and so does $u_j$ in $\alpha$, according to our simulation rule. By now, we have proved the induction basis: for each node in $G_\beta$, in the first time slot, its views in $\alpha$ and $\beta$ are identical.

Assume during the first $i-1$ time slots, for each node in $G_\beta$, its views in $\alpha$ and $\beta$ are identical. We now consider time slot $i$. If a node $u_j$ with identity $j$ chooses to broadcast in $\beta$, then according to the induction hypothesis, it will also choose to broadcast in $\alpha$, with the same content. On the other hand, in time slot $i$, if $u_j$ chooses to listen in $\beta$, then according to the induction hypothesis, $u_j$ will also choose to listen in $\alpha$. This means $R\nsubseteq H_\alpha^i$, implying $R\cap H_\alpha^i=\emptyset$. That is, all nodes with identities in $R$ will choose to listen in time slot $i$ in $\alpha$, which in turn means all nodes in $G_\beta$ will choose to listen in time slot $i$ (according to the induction hypothesis). Thus, $u_j$ will hear silence in time slot $i$ in $\beta$; and so does $u_j$ in $\alpha$, according to our simulation rule. This completes our proof for the claim.

Similarly, we can also prove: for each node in $G_\gamma$, its views in $\alpha$ and $\gamma$ are identical.

Now, imagine an adversary generating a single-hop radio network in the following way: it arbitrarily picks $k$ of $N$ nodes and arbitrarily gives each of these $k$ nodes a unique identity in $[k]$, it then samples a size $c$ subset of $[k]$ uniformly at random and picks the nodes with corresponding identities. Lastly, the adversary flips a fair coin to decide whether to activate all these $c$ nodes, or just two of them (chosen arbitrarily).

Consider the scenario that the adversary chooses $R$ as the size $c$ subset of $[k]$, which happens with probability greater than $1/k^c$ since there are less than $k^c$ size $c$ subsets of $[k]$. In such case, according to our above analysis, by the end of slot $\lg{(k/c)}-1$, nodes cannot distinguish whether there are two nodes in the network or $c$ nodes in the network. If $c> 2\tilde{c}$, then we know by the end of slot $\lg{(k/c)}-1$, there is a $1/2$ chance that the approximation given by $\mathcal{A}$ is incorrect. (Recall we define $\tilde{c}$ in Section \ref{sec-model}.)

To sum up, if $\mathcal{A}$ is a (potentially randomized) distributed algorithm for approximate neighbor counting in single-hop radio network with collision detection, and if $\mathcal{A}$ guarantees to output an estimate by the end of time slot $\lg{(k/c)}-1$, then this estimate is incorrect with probability at least $(1/k^c)\cdot(1/2)$, so long as $c>2\tilde{c}$. Let $(1/k^c)\cdot(1/2)=1/N^\epsilon$ where $\epsilon\geq 1$ is some constant, we know $\lg{(k/c)}-1=(1/c)\cdot\lg{(N^\epsilon/2)}-\lg{c}$. That is, $\lg{(k/c)}-1$ will be $\Theta(\lg{N})$ if $c>2\tilde{c}$ is some constant. By now, we have shown if $\mathcal{A}$ guarantees to output an estimate by the end of time slot $o(\lg{N})$, then this estimate is incorrect with probability at least $1/N^\epsilon$. This proves the first part of the theorem.

\bigskip

Next, we turn our attention to the multi-hop scenario, which generally follows the same high-level strategy as in the single-hop case, but is more involved due to topology changes.

Let $\mathcal{A}$ be an arbitrary (and potentially randomized) distributed algorithm for approximate neighbor counting (in the multi-hop scenario). Assume $G_\alpha$ is a star network in which the designated node $w$ has $k\leq N_\Delta$ neighbors, and each neighbor has a unique identity in $[k]$. We now define $k$ simulations, each focusing on one neighbor $u$ (of the designated node $w$) along with $w$, and runs for $\lg{(k/c)}-1$ time slots. Each simulation follows the following rules. In each time slot, for the designated node $w$, if it chooses to broadcast, then we simulate it broadcasting the content specified by $\mathcal{A}$. If $w$ chooses to listen, then the result depends on whether $u$ broadcasts: $w$ hears collision if $u$ broadcasts, otherwise $w$ hears silence. On the other hand, for node $u$, if it chooses to broadcast, then we simulate it broadcasting the content specified by $\mathcal{A}$. If $u$ chooses to listen, then the result depends on the behavior of $w$: if $w$ broadcasts in this slot, then we simulate $u$ hearing the message sent by $w$; if $w$ does not broadcast in this slot, then we simulate $u$ hearing silence. For each neighbor $u$ of $w$, we call this execution segment of $\mathcal{A}$ as $\alpha_u$. Let $H_\alpha^i$ denote the identities of the neighbors of $w$ that broadcast in slot $i$ (among these $k$ simulations), define set $\mathcal{H}_\alpha=\{H_\alpha^1,H_\alpha^2,\cdots,H_\alpha^{\lg{(k/c)}-1}\}$.

According to Lemma \ref{lemma-comb}, there exists a size $c$ subset $R$ of $[k]$ such that for each $H_\alpha^i\in\mathcal{H}_\alpha$, either $R\subseteq H_\alpha^i$ or $R\cap H_\alpha^i=\emptyset$. Assume $G_\beta$ is a star network in which the neighbors of the designated node are the nodes with identities in $R$. Run $\mathcal{A}$ in $G_\beta$ for $\lg{(k/c)}-1$ time slots, call the resulting execution segment $\beta$. Notice, if $\mathcal{A}$ is randomized, then for each node in $G_\beta$, assume the random bits generated for that node are identical in $\alpha$ and $\beta$. (I.e., for each neighbor $u$ of $w$ in $G_\beta$, same random bits are used in $\alpha_u$ and $\beta$; and for designated node $w$, same random bits are used in all $\alpha_u$ and $\beta$.)

The third star network $G_\gamma$ contains the designated node and two arbitrary nodes with identities in $R$. We run $\mathcal{A}$ in $G_\gamma$ for $\lg{(k/c)}-1$ time slots and call the resulting execution segment $\gamma$. Again, if $\mathcal{A}$ is randomized, then for each node in $G_\gamma$, assume the random bits generated for that node are identical in $\alpha$, $\beta$, and $\gamma$.

We claim, for each neighbor $u$ of the designated node $w$ in $G_\beta$, its views in $\alpha_u$ and $\beta$ are identical. Moreover, for the designated node $w$, its views in all such $\alpha_u$ and $\beta$ are identical. We prove this by a slot to slot induction.

We begin with the first time slot. For the designated node $w$, if it chooses to broadcast in $\beta$, then clearly it will also choose to broadcast in all $\alpha_u$, with the same content. If $w$ chooses to listen in $\beta$, since each node with identity in $R$ will take the same action in the first time slot in (corresponding) $\alpha$ and $\beta$, and since either $R\subseteq H_\alpha^1$ or $R\cap H_\alpha^1=\emptyset$, we know $w$ will hear either silence or collision in $\beta$, and $w$ will hear same thing in all $\alpha_u$ and $\beta$ by our simulation rules. On the other hand, for a neighbor $u$ of the designated node $w$ in $G_\beta$, if $u$ chooses to broadcast in $\beta$, then clearly it will also choose to broadcast in $\alpha_u$, with the same content. If $u$ chooses to listen in $\beta$, then $u$ will also choose to listen in $\alpha_u$. In such case, what $u$ hears in $\alpha_u$ and $\beta$ depend on the behavior of $w$ in $\alpha_u$ and $\beta$. If $w$ broadcasts, then $u$ will hear this message; and if $w$ remains silent, then $u$ will hear silence. Since we have already shown the behavior of $w$ is identical in $\alpha_u$ and $\beta$ in this case, we know $u$'s view will also be identical in $\alpha_u$ and $\beta$ in this case. By now, we have proved the induction basis: for each node in $G_\beta$, in the first time slot, its views in $\alpha$ and $\beta$ are identical.

Assume during the first $i-1$ time slots, for each node in $G_\beta$, its views in $\alpha$ and $\beta$ are identical. We now consider time slot $i$. For the designated node $w$, if it chooses to broadcast in $\beta$, then by the induction hypothesis it will also choose to broadcast in all $\alpha_u$, with the same content. If $w$ chooses to listen in $\beta$, since (by induction hypothesis) each node with identity in $R$ will take same action in time slot $i$ in (corresponding) $\alpha$ and $\beta$, and since either $R\subseteq H_\alpha^i$ or $R\cap H_\alpha^i=\emptyset$, we know $w$ will hear either silence or collision in $\beta$, and $w$ will hear same thing in all $\alpha_u$ and $\beta$ by our simulation rules. On the other hand, for a neighbor $u$ of the designated node $w$ in $G_\beta$, if $u$ chooses to broadcast in $\beta$, then by the induction hypothesis it will also choose to broadcast in $\alpha_u$, with the same content. If $u$ chooses to listen in $\beta$, then the by induction hypothesis $u$ will also choose to listen in $\alpha_u$. In such case, what $u$ hears in $\alpha_u$ and $\beta$ depends on the behavior of $w$ in $\alpha_u$ and $\beta$. If $w$ broadcasts, then $u$ will hear this message; and if $w$ remains silent, then $u$ will hear silence. Since we have shown the behavior of $w$ is identical in $\alpha_u$ and $\beta$ in this case, we know $u$'s view will also be identical in $\alpha_u$ and $\beta$ in this case. This completes our proof for the claim.

Similarly, we can also prove: for each neighbor $u$ of $w$ in $G_\gamma$, its views in $\alpha_u$ and $\gamma$ are identical. Moreover, for the designated node $w$, its views in all such $\alpha_u$ and $\gamma$ are identical.

The remaining proof is almost identical to the single-hop case. Imagine an adversary generates a star network in the following way: it arbitrarily picks $k$ of $N_\Delta$ neighbors and arbitrarily gives each of these $k$ nodes a unique identity in $[k]$, it then samples a size $c$ subset of $[k]$ uniformly at random and picks the nodes with corresponding identities. Lastly, the adversary flips a fair coin to decide whether to activate all these $c$ neighbors, or just two of them (chosen arbitrarily). (The adversary always activates the designated node $w$.)

Now, consider the scenario that the adversary chooses $R$ as the size $c$ subset of $[k]$, which happens with probability at least $1/k^c$ since there are at most $k^c$ size $c$ subsets of $[k]$. In such case, according to our above analysis, by the end of slot $\lg{(k/c)}-1$, the designated node $w$ cannot distinguish whether it has two or $c$ neighbors in the network. If $c$ is a sufficiently large constant, then we know by the end of slot $\lg{(k/c)}-1$, there is a $1/2$ chance that the approximation given by $\mathcal{A}$ is incorrect.

To sum up, if $\mathcal{A}$ is a (potentially randomized) distributed algorithm for approximate neighbor counting in multi-hop radio network with collision detection, and if $\mathcal{A}$ guarantees the designated node outputs an estimate by the end of time slot $\lg{(k/c)}-1$, then this estimate is incorrect with probability at least $(1/k^c)\cdot(1/2)$. Let $(1/k^c)\cdot(1/2)=1/N_\Delta^\epsilon$ where $\epsilon\geq 1$ is some constant, we know $\lg{(k/c)}-1$ will be $\Theta(\log{N_\Delta})$ if $c$ is some constant. I.e., if the designated node always outputs an estimate by the end of time slot $o(\log{N_\Delta})$, then this estimate is incorrect with probability at least $1/N_\Delta^\epsilon$. This completes the proof.
\end{proof}

\section{Omitted Description and Analysis of Upper Bound Results}\label{sec-appx-omit-alg}

\subsection{Omitted proofs for \CountSHnoCDConst}\label{subsec-appx-omit-CountSHnoCDConst}

\begin{lemma}\label{lemma-CountSHnoCDConst-1}
During the execution of \CountSHnoCDConst, if $n\geq 2$, then by the end of iteration $\lg{n}-3$, with at least constant probability, all nodes are still active.
\end{lemma}

\begin{proof}
Notice, for any iteration, if in the first time slot either no node broadcasts or at least two nodes broadcast, then all remaining active nodes will not terminate by the end of this iteration. Thus, to prove the lemma, we need to bound the sum of the probabilities that a single node broadcasts alone in the first time slot, for iterations one to $\lg{n}-3$.

Consider an arbitrary iteration $i\geq 1$, define $\mathcal{E}_i$ to be the event that one node broadcasts alone during the first time slot of iteration $i$. Further define $\mathcal{E}'_i=\mathcal{E}_i~|~\bigwedge_{j=1}^{i-1}\overline{\mathcal{E}_j}$. (I.e., $\mathcal{E}'_i$ is $\mathcal{E}_i$ condition on all nodes are still active at the beginning of iteration $i$.) We know $\mathbb{P}(\mathcal{E}'_i)=n\cdot(1/2^i)\cdot(1-1/2^i)^{n-1}\leq n\cdot(1/2^i)\cdot(1-1/2^i)^{n/2}\leq (n/2^i)\cdot (e^{-0.5})^{n/2^i}$. As a result, $\mathbb{P}(\bigvee_{i=1}^{\lg{n}-3}\mathcal{E}'_i)\leq\sum_{i=1}^{\lg{n}-3}\mathbb{P}(\mathcal{E}'_i)\leq\sum_{j=8}^{n}{j\cdot(e^{-0.5})^j}$. Define $S=\sum_{j=8}^{n}{j\cdot(e^{-0.5})^j}$, we know:

\vspace{-3ex}
\begin{align*}
\left(1-e^{-0.5}\right)S & = 8\cdot\left(e^{-0.5}\right)^8 + \sum_{i=9}^{n}\left(e^{-0.5}\right)^i - n\cdot\left(e^{-0.5}\right)^{n+1} \leq 8\cdot\left(e^{-0.5}\right)^8 + \sum_{i=9}^{n}\left(e^{-0.5}\right)^i \\
& = 8\cdot e^{-4} + \frac{\left(e^{-0.5}\right)^9\cdot\left(1-\left(e^{-0.5}\right)^{n-8}\right)}{1-e^{-0.5}} \leq 8\cdot e^{-4} + \frac{e^{-4.5}}{1-e^{-0.5}}
\end{align*}

Hence $S<0.45$, implying $\mathbb{P}(\bigvee_{i=1}^{\lg{n}-1}\mathcal{E}'_i)<0.45$, and this proves the lemma.
\end{proof}

\begin{lemma}\label{lemma-CountSHnoCDConst-3}
During the execution of \CountSHnoCDConst, for $n\geq 12$, if at the beginning of an iteration $\lg{n}-2\leq i\leq\lg{n}$ all nodes are active, then by the end of this iteration, with at least constant probability, all nodes are either active or have terminated.
\end{lemma}

\begin{proof}
Consider an iteration $\lg{n}-2\leq i\leq\lg{n}$ in which all nodes are active at the beginning of it. If by the end of it all nodes are still active, then we are done. So, assume a node $u$ decides to terminate by the end of this iteration. We need to show, in such case, with at least some constant probability, all other nodes must have also decided to terminate by the end of iteration $i$ as well.

Observe that $u$ will decide to terminate if one of the following two events happen: (a) $u$ hears a \textsf{beacon} message in the first time slot; or (b) $u$ hears a \textsf{stop} message in the second time slot. (Notice, according to our protocol, (a) and (b) cannot both happen.)

If it is the case that event (a) happens to $u$, then there must exist a node $v\neq u$ that broadcasts alone in the first time slot. Therefore, we know every node except $v$ will terminate after iteration $i$. As for node $v$, it will listen in the second time slot. Moreover, it will hear a \textsf{stop} message with probability at least $(n-1)\cdot(1/(2^i-1))\cdot(1-1/(2^i-1))^{n-2}\geq (n-1)\cdot(1/n)\cdot(1-1/(n/4-1))^{n-2}>1/1000$.

On the other hand, if it is the case that event (b) happens to $u$, then it must be the case that $u$ broadcasts alone in the first time slot. Thus, we know every node except $u$ will terminate after iteration $i$. Moreover, since $u$ hears a \textsf{stop} message in the second time slot, it will also terminate after iteration $i$.
\end{proof}

\begin{lemma}\label{lemma-CountSHnoCDConst-2}
During the execution of \CountSHnoCDConst, if $n\geq 3$ and all nodes are active at the beginning of iteration $\lg{n}$, then by the end of iteration $\lg{n}$, with at least constant probability, all nodes will terminate.
\end{lemma}

\begin{proof}
Since all nodes are active at the beginning of iteration $\lg{n}$, we know in the first time slot in this iteration, the probability that exactly one node broadcasts is $n\cdot(1/n)\cdot(1-1/n)^{n-1}\geq (1-1/n)^n\geq 1/4$. Assume this event indeed happens, and the node that broadcasts is $u$. In such case, in the second time slot in this iteration, the probability that $u$ will hear a \textsf{stop} message is $(n-1)\cdot(1/(n-1))\cdot(1-1/(n-1))^{n-2}\geq(1-1/(n-1))^{n-1}\geq 1/5$. Therefore, we can conclude, if all nodes are active at the beginning of iteration $\lg{n}$, then they will all terminate by the end of this iteration, with at least constant probability.
\end{proof}

\subsection{Omitted description and analysis for \CountSHnoCDHigh}\label{subsec-appx-omit-CountSHnoCDHigh}

We first give the complete description of the second and third phases. In each iteration $i$, both the second phase and third phase are of length $\Theta(i)$. More specifically, in the second phase, for a node $u$, its behavior depends on whether it has obtained its private estimate. If it has not obtained a private estimate, then it simply listens in all slots. If, however, $u$ has already obtained a private estimate, then in each time slot, $u$ will broadcast an \textsf{informed} message with probability $1/2^i$, and listen otherwise. By the end of the second phase, if $u$ has ever heard an \textsf{informed} message, it will terminate by the end of this iteration, and commit to $2^{i+2}$ as its final estimate for $n$. Finally, the third phase is used to deal with the case in which one ``unlucky'' node successfully broadcasts an \textsf{informed} message during phase two (thus terminate all other nodes), but never gets the chance to successfully receive an \textsf{informed} message (thus cannot terminate along with other nodes). Specifically, in the third phase of iteration $i$, if a node has already committed to a final estimate, then in each time slot, it will broadcast a \textsf{stop} message with probability $1/2^i$. On the other hand, if a node has not decided its final estimate yet, then it will simply listen in each time slot. Moreover, if the node ever hears a \textsf{stop} message, then it will terminate by the end of this iteration and use $2^{i+2}$ as its final estimate.

Pseudocode of \CountSHnoCDHigh is shown in Figure \ref{fig-alg-CountSHnoCDHigh}.

\begin{figure}[!t]
\hrule
\vspace{1ex}\textbf{Pseudocode of \CountSHnoCDHigh executed at node $u$:}\vspace{1ex}
\hrule
\begin{small}
\begin{algorithmic}[1]
\State $status\gets uninformed, i\gets 1$.
\While {($true$)}
	\Statex \hspace{3ex}$\triangleright$ \textsc{Phase I}
	\State $count\_listen\gets 0, count\_msg\gets 0$.
	\For {($i\gets 1$ to $\Theta(i)$)}
		\If {($\texttt{random}(1,2^i)==1$)} \Comment $\texttt{random}(x,y)$ returns a random integer in range $[x,y]$.
			\State $\texttt{broadcast}(\langle\textsf{beacon}\rangle)$.
		\Else
			\State $msg\gets\texttt{listen}()$.
			\State $count\_listen\gets count\_listen+1$.
			\If {($msg==\langle\textsf{beacon}\rangle$)}
				\State $count\_msg\gets count\_msg+1$.
			\EndIf
		\EndIf
	\EndFor
	\If {($status==uninformed$ \textbf{and} $count\_listen\neq 0$ \textbf{and} $\frac{count\_msg}{count\_listen}\geq 1/2e$)}
		\State $status\gets informed$.
		\State $est\_private\gets 2^i$.
	\EndIf
	\Statex \hspace{3ex}$\triangleright$ \textsc{Phase II}
	\For {($i\gets 1$ to $\Theta(i)$)}
		\If {($status==informed$ \textbf{and} $\texttt{random}(1, 2^i)==1$)}
			\State $\texttt{broadcast}(\langle\textsf{informed}\rangle)$.
		\Else
			\State $msg\gets\texttt{listen}()$.
			\If {($msg==\langle\textsf{informed}\rangle$)}
				\State $status\gets stopped$.
				\State $est\gets 2^{i+2}$. \Comment $est$ is the final estimate for $n$.
			\EndIf
		\EndIf
	\EndFor
	\Statex \hspace{3ex}$\triangleright$ \textsc{Phase III}
	\For {($i\gets 1$ to $\Theta(i)$)}
		\If {($status==stopped$ \textbf{and} $\texttt{random}(1,2^i)==1$)}
			\State $\texttt{broadcast}(\langle\textsf{stop}\rangle)$.
		\ElsIf {($status\neq stopped$)}
			\State $msg\gets\texttt{listen}()$.
			\If {($msg==\langle\textsf{stop}\rangle$)}
				\State $status\gets stopped$.
				\State $est\gets 2^{i+2}$.
			\EndIf
		\EndIf
	\EndFor
	\If {($status==stopped$)}
		\State \textbf{return} $est$.
	\Else
		\State $i\gets i+1$.
	\EndIf
\EndWhile
\end{algorithmic}
\end{small}
\hrule\vspace{1ex}
\caption{Pseudocode of the \CountSHnoCDHigh algorithm.}\label{fig-alg-CountSHnoCDHigh}
\vspace{-3ex}
\end{figure}

We now give detailed analysis of \CountSHnoCDHigh. To begin with, we show that by the end of iteration $\lg{(n/(a\ln{n}))}$, no node has obtained its private or final estimate (thus all nodes are still active), with high probability in $n$. Here, $a$ is a sufficiently large constant.

\begin{lemma}\label{lemma-CountSHnoCDHigh-1}
During the execution of \CountSHnoCDHigh, if $n\geq 2$, then by the end of iteration $\lg{(n/(a\ln{n}))}$, with high probability in $n$, no node has obtained its private or final estimate for $n$. Here, $a$ is a sufficiently large constant.
\end{lemma}

\begin{proof}
Notice, according to our protocol description, if at the beginning of an iteration no node has obtained its private or final estimate (which implies all nodes are active), and if during the first phase of this iteration no node ever hears a \textsf{beacon} message, then by the end of this iteration, still no node will have private or final estimate (thus will remain active).

Now, consider an iteration $i\leq \lg{(n/(a\ln{n}))}$ in which all nodes are still active, and a time slot in the first phase of iteration $i$. According to the protocol, in such a time slot, the probability that a fixed node $u$ will hear a \textsf{beacon} message sent by some neighbor is at most $n\cdot(1/2^i)\cdot(1-1/2^i)^{n-1}\leq n\cdot(1-1/2^i)^{n-1}\leq n\cdot(1-(a\ln{n})/n)^{n-1}\leq n\cdot(1-(a\ln{n})/n)^{n/2}\leq n\cdot e^{-((a\ln{n})/n)\cdot(n/2)}=n\cdot e^{-(a/2)\cdot\ln{n}}=n^{-a/2+1}$. Take a union bound over the $n$ nodes, a union bound over the $\Theta(i)=O(\lg{n})$ time slots within one iteration, and another union bound over the $\lg{(n/(a\ln{n}))}=O(\lg{n})$ iterations, the lemma follows.
\end{proof}

We then focus on iterations $\lg{(n/(a\ln{n}))}$ to $\lg{n}-3$, and show that nodes will not obtain private or final estimates during these iterations as well.

\begin{lemma}\label{lemma-CountSHnoCDHigh-2}
During the execution of \CountSHnoCDHigh, assume no node has obtained its private or final estimate for $n$ by the end of iteration $\lg{(n/(a\ln{n}))}$. In such case, if $n\geq 7$, then by the end of iteration $\lg{n}-3$, with high probability in $n$, still no node has obtained its private or final estimate for $n$. Here, $a$ is a sufficiently large constant.
\end{lemma}

\begin{proof}
Notice, according to our protocol description, if at the beginning of an iteration no node has obtained its private or final estimate (which implies all nodes are active), and if during the first phase of this iteration no node observes at least $1/2e$ fraction of clear message slots (among all listening slots), then by the end of this iteration, still no node will have private or final estimate (thus will remain active).

Now, consider a fixed iteration $i$ in which all nodes are active and no node has obtained its private or final estimate at the beginning of iteration $i$. Here, $\lg{(n/(a\ln{n}))}\leq i\leq \lg{n}-3$. Let $m_i=\Theta(i)$ denote the length of the first phase of iteration $i$. According to the protocol, for a given node $u$, in each time slot in phase one, the node will choose to listen with probability at least $1/2$. Apply a standard Chernoff bound~\cite{mitzenmacher05}, we know $u$ will choose to listen in at least $\Omega(\lg{n})$ time slots, with high probability in $n$.

Assume node $u$ in fact chooses to listen in $l_i=\Omega(\lg{n})$ time slots during phase one. Let $X_i$ be a random variable denoting the number of slots in which $u$ hears a \textsf{beacon} message among the $l_i$ listening slots. Let $\mu_i=\mathbb{E}(X_i)$. We know $\mu_i=l_i\cdot (n-1)\cdot(1/2^i)\cdot(1-1/2^i)^{n-2}$. When $\lg{(n/(a\ln{n}))}\leq i\leq\lg{n}-3$, we can further bound $\mu_i$ from above: $\mu_i\leq l_i\cdot(n-1)\cdot(8/n)\cdot(1-8/n)^{n-2}\leq l_i\cdot n\cdot(8/n)\cdot(1-8/n)^{0.7n}\leq 8l_i\cdot e^{-(8/n)\cdot 0.7n}=8e^{-5.6}\cdot l_i$. Notice, the first inequality holds since $\mu_i$ is non-decreasing when $\lg{(n/(a\ln{n}))}\leq i\leq\lg{n}-3$.

In the remaining proof, we will use the following Chernoff-style bound which is easy to derive based on the original Chernoff bounds shown in \cite{mitzenmacher05}.

\begin{claim}\label{claim-chernoff-bound-variant}
Let $X_1,\cdots,X_n$ be independent indicator random variables such that $\mathbb{P}(X_i=1)=p_i$. Let $X=\sum_{i=1}^{n}{X_i}$ and $\mu=\mathbb{E}(X)$. Then, for $\delta>1$, we have $\mathbb{P}(X\geq \delta\mu)<(e/\delta)^{\delta\mu}$.
\end{claim}

Define $t_i=24e^{-5.6}\cdot l_i$ and $\delta_i=t_i/\mu_i$. According to Claim \ref{claim-chernoff-bound-variant}, we know $\mathbb{P}(X_i\geq t_i)<(e\mu_i/t_i)^{t_i}=(e/t_i)^{t_i}\cdot{\mu_i}^{t_i}\leq (e/t_i)^{t_i}\cdot{(8e^{-5.6}\cdot l_i)}^{t_i}=(e/3)^{t_i}\leq (e/3)^{24e^{-5.6}\cdot l_i}=(e/3)^{\Omega(\lg{n})}=1/{n}^{\Omega(1)}$. Thus, we know with high probability in $n$, node $u$ will observe at most $24e^{-5.6}$ fraction of clear message slots among all $l_i$ listening slots during phase one. Since $24e^{-5.6}<1/2e$, take a union bound over the $n$ nodes, and another union bound over iterations $\lg{(n/(a\ln{n}))}$ to $\lg{n}-3$, the lemma follows.
\end{proof}

Starting from iteration $\lg{n}-2$, nodes will begin to obtain private estimates, and may terminate with non-trivial probability. Nonetheless, our algorithm guarantees either no node terminates, or all nodes terminate simultaneously with same (final) estimate:

\begin{lemma}\label{lemma-CountSHnoCDHigh-3}
During the execution of \CountSHnoCDHigh, for $n\geq 5$, if at the beginning of an iteration $\lg{n}-2\leq i\leq\lg{n}$ all nodes are active, then by the end of this iteration, with high probability in $n$, either all nodes are still active, or all nodes have terminated with final estimate $2^{i+2}$.
\end{lemma}

\begin{proof}
Consider an iteration $\lg{n}-2\leq i\leq\lg{n}$ in which all nodes are active at the beginning of it. If by the end of it all nodes are still active, then we are done. So, assume a node $u$ decides to terminate by the end of this iteration. We need to show, with high probability in $n$, all other nodes must have also decided to terminate by the end of iteration $i$.

Observe that $u$ will decide to terminate if one of the following two events happen: (a) $u$ hears an \textsf{informed} message during the second phase; or (b) $u$ hears a \textsf{stop} message during the third phase. (Notice, according to our protocol, (a) and (b) cannot both happen.)

If it is the case that event (a) happens to $u$, then there must exist a node $v\neq u$ and a slot in the second phase such that in that slot $v$ broadcasts alone while all other nodes listen. Therefore, we know by the end of the second phase, every node except $v$ must have decided to terminate after iteration $i$, with $2^{i+2}$ begin their final estimate. As for node $v$, if it has not decided to terminate after iteration $i$ by the end of phase two, then it will do so during phase three. More specifically, according to protocol description, in each slot during phase three, each node but $v$ will broadcast a \textsf{stop} message with probability $1/2^i$, while $v$ listens. Thus, in each slot during phase three, $v$ will hear a \textsf{stop} message with probability at least $(n-1)\cdot(1/2^i)\cdot(1-1/2^i)^{n-2}\geq (n-1)\cdot(1/n)\cdot(1-4/n)^{n-2}>1/200$. Since phase three contains $\Theta(i)=\Omega(\lg{n})$ independent time slots, we know by the end of phase three, with high probability in $n$, node $v$ must have decided to terminate as well (with $2^{i+2}$ being its final estimate).

On the other hand, if it is the case that event (b) happens to $u$, then we know there must exist a node $v$ that has already decided to terminate by the end of phase two of iteration $i$ (otherwise no node would broadcast in phase three). This further implies during phase two, node $v$ must have heard an \textsf{informed} message. Hence, there must exist a node $u'\neq v$ and a slot in the second phase such that in that slot $u'$ broadcasts alone while all other nodes listen. Therefore, we know by the end of the second phase, every node except $u'$ must have decided to terminate after iteration $i$, with $2^{i+2}$ being their final estimate. Notice, since $u$ receives a \textsf{stop} message during phase three, we know it must be the case that $u'=u$. By now, we can conclude, by the end of phase three, all nodes will terminate with $2^{i+2}$ being their final estimate.
\end{proof}

To complete the correctness proof of \CountSHnoCDHigh, we show by the end of iteration $\lg{n}$, all nodes must have terminated and output their (identical) estimates.

\begin{lemma}\label{lemma-CountSHnoCDHigh-4}
During the execution of \CountSHnoCDHigh, if $n\geq 10$, then by the end of iteration $\lg{n}$ all nodes must have terminated, with high probability in $n$.
\end{lemma}

\begin{proof}
Due to Lemma \ref{lemma-CountSHnoCDHigh-1} and \ref{lemma-CountSHnoCDHigh-2}, we know all nodes are still active by the end of iteration $\lg{n}-3$, with high probability in $n$. If by the end of iteration $\lg{n}-1$ all nodes have already terminated, then we are done. Otherwise, assume at the beginning of iteration $\lg{n}$, at least one node is still active. Due to Lemma \ref{lemma-CountSHnoCDHigh-3}, we know in such case, all nodes must be active at the beginning of iteration $\lg{n}$, with high probability in $n$.

Assume indeed all nodes are active at the beginning of iteration $\lg{n}$. Let $m=\Theta(\lg{n})$ denote the length of the first phase of iteration $\lg{n}$. According to the protocol, for a given node $u$, in each time slot in phase one, the node will choose to listen with probability at least $1/2$. Apply a standard Chernoff bound, we know $u$ will choose to listen in at least $\Omega(\lg{n})$ time slots, with high probability in $n$.

Assume $u$ in fact chooses to listen in $l=\Omega(\lg{n})$ time slots during phase one. Let $X$ be a random variable denoting the number of slots in which $u$ hears a \textsf{beacon} message among the $l$ listening slots. We know $\mathbb{E}(X)=l\cdot (n-1)\cdot(1/n)\cdot(1-1/n)^{n-2}>0.9l\cdot (1-1/n)^{n}>0.9l\cdot 4^{-(1/n)\cdot n}=0.9l/4$. Apply a Chernoff bound we know $X$ will be at least $l/5$, with high probability in $n$. Since $1/2e<1/5$, we know $u$ must have obtained its private estimate by the end of phase one, with high probability in $n$. Take a union bound over all the $n$ nodes, we know this claim holds true for them too.

Assume in iteration $\lg{n}$, all nodes are active and each node has obtained its private estimate by the end of phase one, we now focus on phase two. Fix a node $u$, according to the protocol, we know in each slot in phase two, $u$ will hear an \textsf{informed} message with probability at least $(1-1/n)\cdot(n-1)\cdot(1/n)\cdot(1-1/n)^{n-2}>0.9\cdot (1-1/n)^{n}>0.9\cdot 4^{-(1/n)\cdot n}>1/5$. Since phase two contains $\Theta(\lg{n})$ independent time slots, we knot by the end of phase two, $u$ must have heard an \textsf{informed} message and thus decides to terminate by the end of iteration $\lg{n}$, with high probability in $n$. Take a union bound over all the $n$ nodes, we know this claim holds true for them too. By now, we have proved the lemma.
\end{proof}

Combine the above lemmas and we can immediately have Theorem \ref{thm-CountSHnoCDHigh}.

\subsection{Omitted description and analysis for \EstUpperSH}\label{subsec-appx-omit-EstUpperSH}

\EstUpperSH contains multiple iterations, and is quite similar to \CountSHnoCDConst. In the $i$\textsuperscript{th} iteration, nodes assume $\lg{n}=2^i$, and then verify the correctness of the estimate. In case the estimate is correct, all nodes terminate simultaneously, otherwise they continue into the next iteration. More specifically, the $i$\textsuperscript{th} iteration contains two slots. In the first slot, each node broadcasts a \textsf{beacon} message with probability $1/2^{2^i}$, and otherwise listens. If a node listens in the first time slot and hears silence or a \textsf{beacon} message, it will broadcast a \textsf{stop} message in the second time slot and then terminate, with $2^{i+1}$ being its estimate for $\lg{n}$. On the other hand, if a node chooses to broadcast in the first time slot, then it will listen in the second time slot. Moreover, in case it hears noise or a \textsf{stop} message in the second time slot, it will terminate and use $2^{i+1}$ as its estimate of $\lg{n}$.

To prove the correctness of \EstUpperSH, we first show that before iteration $\lfloor\lg\lg{n}\rfloor$, nodes will not terminate. (Recall we assume $n$ is a power of two, thus $\lg{n}$ is an integer, but $\lg\lg{n}$ is not necessarily an integer.) Intuitively, this is because in these iterations, estimates are too small and listening nodes are likely to always hear noise in the first time slot.

\begin{lemma}\label{lemma-EstUpperSH-1}
During the execution of \EstUpperSH, for sufficiently large $n$, by the end of iteration $\lfloor\lg\lg{n}\rfloor-1$, all nodes are still active, with high probability in $n$.
\end{lemma}

\begin{proof}
Consider an arbitrary iteration $i$ where $1\leq i\leq \lfloor\lg\lg{n}\rfloor-1$. If all nodes are active at the beginning of iteration $i$, then according to our protocol, in the first slot in that iteration, the probability that at most one node will broadcast is $n\cdot(1/2^{2^i})\cdot(1-1/2^{2^i})^{n-1} + (1-1/2^{2^i})^{n}\leq 1/n^{O(1)}$. (See Figure \ref{fig-eqn-array-lemma-EstUpperSH-1} for details.)

\begin{figure}[!t]
\begin{align*}
& n\cdot(1/2^{2^i})\cdot(1-1/2^{2^i})^{n-1} + (1-1/2^{2^i})^{n} \\
& \leq n\cdot(1/2^{2^{\lfloor\lg\lg{n}\rfloor-1}})\cdot(1-1/2^{2^{\lfloor\lg\lg{n}\rfloor-1}})^{n-1} + (1-1/2^{2^{\lfloor\lg\lg{n}\rfloor-1}})^{n} \\
& \leq n\cdot(1/2^{2^{\lg\lg{n}-1}})\cdot(1-1/2^{2^{\lg\lg{n}-1}})^{n-1} + (1-1/2^{2^{\lg\lg{n}-1}})^{n} \\
& = \sqrt{n}\cdot(1-1/\sqrt{n})^{n-1} + (1-1/\sqrt{n})^{n} \\
& \leq \sqrt{n}\cdot(1-1/\sqrt{n})^{0.98n} + e^{-\sqrt{n}} \leq \sqrt{n}\cdot e^{-0.98\sqrt{n}} + e^{-0.98\sqrt{n}}\\
& \leq e^{0.3\sqrt{n}}\cdot e^{-0.98\sqrt{n}}=e^{-\Omega(\ln{n})}=1/{n}^{O(1)}
\end{align*}
\vspace{-5ex}\caption{}\label{fig-eqn-array-lemma-EstUpperSH-1}
\end{figure}

This implies all nodes will still be active after iteration $i$, with probability at least $1-1/{n}^{O(1)}$. Take a union bound over all the $\lfloor\lg\lg{n}\rfloor-1$ iterations, the lemma is proved.
\end{proof}

Next, we show by iteration $\lfloor\lg\lg{n}\rfloor+O(1)$, all nodes will terminate.

\begin{lemma}\label{lemma-EstUpperSH-2}
During the execution of \EstUpperSH, if $n\geq 2$ and nodes have not terminated by the end of iteration $\lfloor\lg\lg{n}\rfloor-1$, then by the end of iteration $\lfloor\lg\lg{n}\rfloor+1$, all nodes must have terminated with at least some constant probability. Moreover, by the end of iteration $\lfloor\lg\lg{n}\rfloor+k$, all nodes must have terminated with probability at least $1-2/({n}^{2^{k-1}-1})$. Here, $k\geq 2$ is an integer.
\end{lemma}

\begin{proof}
Let $c\geq 1$ be an integer. Notice, it is easy to see our protocol guarantees by the end of an iteration, either no node terminates, or all nodes terminate. If all nodes have already terminated by the end of iteration $\lfloor\lg\lg{n}\rfloor+(c-1)$, then we are done. So, assume all nodes are still active at the beginning of iteration $\lfloor\lg\lg{n}\rfloor+c$. In such case, in the first slot in this iteration, the probability that all $n$ neighbors of $u$ will remain silent is $(1-1/2^{2^{\lfloor\lg\lg{n}\rfloor+c}})^{n}> (1-1/2^{2^{\lg\lg{n}+(c-1)}})^{n}=(1-1/{n}^{2^{c-1}})^{n}\geq (2e)^{-{n}/{n}^{2^{c-1}}}= (2e)^{-{n}^{1-2^{c-1}}}$.

When $c=1$, $(2e)^{-{n}^{1-2^{c-1}}}=1/2e$ is a constant. On the other hand, when $c\geq 2$, $(2e)^{-{n}^{1-2^{c-1}}}\geq e^{-(\ln(2e))\cdot{n}^{1-2^{c-1}}}\geq 1-(\ln{(2e)})/({n}^{2^{c-1}-1})> 1-2/({n}^{2^{c-1}-1})$.
\end{proof}

Clearly, the above two lemmas imply the correctness of \EstUpperSH:

\begin{theorem}\label{thm-EstUpperSH}
In single-hop radio networks, when collision detection is available, \EstUpperSH guarantees the following properties: (a) all nodes terminate simultaneously (if they ever terminate); (b) with at least some constant probability, all nodes get the same estimate of $\lg{n}$ in range $[2^{\lfloor\lg\lg{n}\rfloor+1}, 2^{\lfloor\lg\lg{n}\rfloor+2}]$ in $O(\lg\lg{n})$ time; and (b) with probability at least $1-1/{n}^{k-2}$, all nodes get the same estimate of $\lg{n}$ in range $[2^{\lfloor\lg\lg{n}\rfloor+1}, 2^{\lfloor\lg\lg{n}\rfloor+k}]$ in $O(\lg\lg{n})$ time. Here, $k\geq 4$ is an integer.
\end{theorem}

\subsection{Omitted description and analysis for \CountSHCDConst}\label{subsec-appx-omit-CountSHCDConst}

We first give a complete and detailed description of \CountSHCDConst.

\CountSHCDConst contains multiple iterations, each of which contains four time slots. In each iteration $i$, all nodes have a lower bound $a_i$ and an upper bound $b_i$, and will test whether the median $m_i=\lfloor (a_i+b_i)/2\rfloor$ is close to $\lg{n}$ or not. (Initially, $a_1$ is set to one, and $b_1$ is set to the estimate of $\lg{n}$ returned by \EstUpperSH.) In the first time slot in iteration $i$, each node will choose to broadcast a \textsf{beacon} message with probability $1/2^{m_i}$, and listen otherwise. If a node listens in the first time slot and hears silence, it will set $b_{i+1}$ to $m_i-1$; if it hears noise, it will set $a_{i+1}$ to $m_i+1$; and if it hears a \textsf{beacon} message, it will terminate by the end of this iteration with $2^{m_i+1}$ being its estimate of $n$. The other three time slots in each iteration $i$ allow nodes that have chosen to broadcast in the first time slot to learn the status of the channel, with the help of nodes that have chosen to listen in the first time slot. More specifically, for each node that have chosen to listen in the first time slot: if it heard silence, then it will broadcast an \textsf{over-est} message in the second time slot; if it heard noise, then it will broadcast an \textsf{under-est} message in the third time slot; and if it heard a \textsf{beacon} message, then it will broadcast a \textsf{stop} message in the last time slot. On the other hand, for each node that have chosen to broadcast in the first time slot, it will listen in the next three time slots. Moreover, if it hears noise or a message in the second time slot, it will set $b_{i+1}$ to $m_i-1$; if it hears noise or a message in the third time slot, it will set $a_{i+1}$ to $m_i+1$; and if it hears noise or a message in the last time slot, it will terminate with $2^{m_i+1}$ being its estimate of $n$. Finally, we note that, if at the beginning of some iteration $i$, a node finds $a_i>b_i$, then it will terminate without obtaining an estimate of $n$.

The complete pseudocode of \CountSHCDConst is given in Figure \ref{fig-alg-CountSHCDConst}.

\begin{figure}[!t]
\hrule
\vspace{1ex}\textbf{Pseudocode of \CountSHCDConst executed at node $u$:}\vspace{1ex}
\hrule
\begin{small}
\begin{algorithmic}[1]
\State $a\gets 1, b\gets \EstUpperSH()$.
\While {($true$)}
	\If {($a>b$)}\ \textbf{abort}.\EndIf
	\State $m\gets\lfloor(a+b)/2\rfloor$.
	\If {($\texttt{random}(1,2^m)==1$)}\ $role\gets bcst$ \algorithmicelse\ $role\gets listen$.\EndIf
	\If {($role==listen$)}
		\State $msg\gets\texttt{listen}()$. \Comment Listen in first time slot.
		\If {($msg==silence$)} \Comment Heard silence in first time slot.
			\State $\texttt{broadcast}(\langle\textsf{over-est}\rangle); \texttt{idle}(); \texttt{idle}()$. \Comment Only broadcast \textsf{over-est} in second time slot.
			\State $b\gets m-1$.
		\ElsIf {($msg==noise$)} \Comment Heard noise in first time slot.
			\State $\texttt{idle}(); \texttt{broadcast}(\langle\textsf{under-est}\rangle); \texttt{idle}()$. \Comment Only broadcast \textsf{under-est} in third time slot.
			\State $a\gets m+1$.
		\Else \Comment Heard \textsf{beacon} message in first time slot.
			\State $\texttt{idle}(); \texttt{idle}(); \texttt{broadcast}(\langle\textsf{stop}\rangle)$. \Comment Only broadcast \textsf{stop} in fourth time slot.
			\State \textbf{return} $2^{m+1}$.
		\EndIf
	\Else
		\State $\texttt{broadcast}(\langle\textsf{beacon}\rangle)$. \Comment Broadcast \textsf{beacon} in first time slot.
		\State $msg1\gets\texttt{listen}()$. \Comment Listen in remaining three time slots.
		\State $msg2\gets\texttt{listen}()$.
		\State $msg3\gets\texttt{listen}()$.
		\If {($msg1\neq silence$)}
			\State $b\gets m-1$.
		\ElsIf {($msg2\neq silence$)}
			\State $a\gets m+1$.
		\Else \Comment Heard noise or \textsf{stop} message in fourth time slot.
			\State \textbf{return} $2^{m+1}$.
		\EndIf
	\EndIf
\EndWhile
\end{algorithmic}
\end{small}
\hrule\vspace{1ex}
\caption{Pseudocode of the \CountSHCDConst algorithm.}\label{fig-alg-CountSHCDConst}
\vspace{-3ex}
\end{figure}

Before giving detailed analysis for \CountSHCDConst, we first show that the last three slots within each iteration ensure all nodes are ``tightly synchronized'' and always have same value of $a_i$ and $b_i$.

\begin{lemma}\label{lemma-CountSHCDConst-1}
During the execution of \CountSHCDConst, by the end of each iteration: (a) all nodes have same value of $a_i$ and $b_i$; and (b) either all nodes decide to terminate, or all nodes decide to continue.
\end{lemma}

\begin{proof}
Assume all nodes are active at the beginning of iteration $i$, and have same $a_i$ and $b_i$.

If a node $u$ changes $b_{i+1}$ to $m_i-1$, then it must have chosen to listen in the first time slot and heard nothing. This implies all nodes have chosen to listen in the first time slot, which in turn implies all nodes have heard nothing in the first time slot. Thus, all nodes must have changed $b_{i+1}$ to $m_i-1$.

If a node $u$ changes $a_{i+1}$ to $m_i+1$, then it must have chosen to listen in the first time slot and heard noise. This implies all nodes that have chosen to listen in the first time slot must also have heard noise and set $a_{i+1}$ to $m_i+1$. On the other hand, all nodes that have chosen to broadcast in the first time slot will listen in the third time slot. Moreover, in this time slot, these nodes must have heard something (noise or \textsf{under-est}) since at least $u$ will broadcast \textsf{under-est}. Thus, these nodes will also set $a_{i+1}$ to $m_i+1$.

We continue to show all nodes will terminate simultaneously. If by the end of iteration $i$ all nodes decide to continue, then we are fine. If some node $u$ decides to terminate, then we claim all other nodes must have decided to terminate as well. To see this, notice that $u$ will terminate in case one of the three following events happens: (a) $u$ finds $a_i>b_i$; (b) $u$ hears a \textsf{beacon} message in the first time slot; or (c) $u$ hears noise or a \textsf{stop} message in the fourth time slot. (Notice, at most one of these events can happen to $u$.)

In case (a) happens to $u$, by our above analysis we know all nodes find $a_i>b_i$.

In case (b) happens, there must exist one node $v\neq u$ that broadcasts alone in the first time slot. Thus, all nodes but $v$ will decide to terminate by the end of iteration $i$. As for $v$, since it broadcast in the first time slot, it will listen in the fourth time slot. Moreover, in that slot, it must have heard a \textsf{stop} message or noise, since at least $u$ will broadcast \textsf{stop} in that time slot. Thus, $v$ will terminate by the end of iteration $i$ too.

And finally, in case (c) happens, $u$ must have chosen to broadcast in the first time slot. Thus, all nodes that have chosen to listen in the first time slot will terminate by the end of iteration $i$. As for the nodes that have chosen to broadcast in the first time slot, they will listen in the fourth time slot. Since $u$ hears something in the fourth time slot, these nodes must have heard something in the fourth time slot as well. Thus, these nodes will also choose to terminate by the end of iteration $i$.
\end{proof}

We are now ready to present the detailed correctness proof. To begin with, we show that nodes can make the correct decision whenever $m_i$ is in range $[1,\lg{n}-2]$.

\begin{lemma}\label{lemma-CountSHCDConst-2}
During the execution of \CountSHCDConst, if $n\geq 10$, then with probability at least $0.75$, during iterations in which all nodes are active and $m_i$ is in range $[1, \lg{n}-2]$, all nodes will adjust $a_{i+1}$ to $m_i+1$.
\end{lemma}

\begin{proof}
Consider an iteration $i$ in which all nodes are active and $m_i\in [1, \lg{n}-2]$, let $\mathcal{E}_{m_i}$ be the (bad) event that in the first time slot in that iteration at most one node choose to broadcast. According to the protocol, we know $\mathbb{P}(\mathcal{E}_{m_i})=(1-1/2^{m_i})^{n}+n\cdot(1/2^{m_i})\cdot(1-1/2^{m_i})^{n-1}\leq (1-1/2^{m_i})^{n}+n\cdot(1/2^{m_i})\cdot(1-1/2^{m_i})^{0.9n}\leq e^{-n/2^{m_i}}+(n/2^{m_i})\cdot e^{-0.9n/2^{m_i}}$.

We now bound the sum of $\mathbb{P}(\mathcal{E}_{m_i})$ when $1\leq {m_i}\leq \lg{n}-2$:

\vspace{-3ex}
\begin{align*}
\sum_{{m_i}=1}^{\lg{n}-2}{\mathbb{P}(\mathcal{E}_{m_i})} & \leq \sum_{{m_i}=1}^{\lg{n}-2}{\left(e^{-n/2^{m_i}}+(n/2^{m_i})\cdot e^{-0.9n/2^{m_i}}\right)} < \sum_{k=4}^{n}{\left(e^{-k}+k\cdot e^{-0.9k}\right)}
\end{align*}

It is easy to show $\sum_{k=4}^{n}{e^{-k}}<e^{-4}/(1-e^{-1})<0.03$. On the other hand, it is also not hard to prove $\sum_{k=4}^{n}{k\cdot e^{-0.9k}}<(4(e^{-0.9})^4-3(e^{-0.9})^5)/(1-e^{-0.9})^2<0.22$. Therefore, $\sum_{{m_i}=1}^{\lg{n}-2}{\mathbb{P}(\mathcal{E}_{m_i})}<0.25$.

Notice, during the execution of \CountSHCDConst, for each ${m_i}$ where $1\leq {m_i}\leq \lg{n}-2$, that value of $m_i$ will be used in at most one iteration (due to the properties of binary search). Since $\sum_{{m_i}=1}^{\lg{n}-2}{\mathbb{P}(\mathcal{E}_{m_i})}<0.25$, we know $\mathbb{P}(\bigwedge_{{m_i}=1}^{\lg{n}-2}{\overline{\mathcal{E}_{m_i}}})>0.75$.

As a result, the lemma is proved.
\end{proof}

We then show whenever $m_i$ is in range $[\lg{n}+2, c\cdot \lg{n}]$, all nodes will also make the correct decision of decreasing the upper bound to $m_i-1$. Here, $c\cdot \lg{n}$ is the upper bound of $\lg{n}$ returned by \EstUpperSH. Particularly, $c$ is some positive integer.

\begin{lemma}\label{lemma-CountSHCDConst-3}
During the execution of \CountSHCDConst, if $n\geq 2$, then with probability at least $0.45$, during iterations in which all nodes are active and $m_i$ is in range $[\lg{n}+2, c\cdot \lg{n}]$, all nodes will adjust $b_{i+1}$ to $m_i-1$. Here, $c$ is some positive integer.
\end{lemma}

\begin{proof}
Consider an iteration $i$ in which all nodes are active and $m_i\in [\lg{n}+2, c\cdot \lg{n}]$, let $\mathcal{E}_{m_i}$ be the (bad) event that in the first time slot in that iteration at least one node choose to broadcast. We know $\mathbb{P}(\mathcal{E}_{m_i})=1-(1-1/2^{m_i})^{n}\leq 1-(\beta e)^{-n/2^{m_i}}=1-e^{-(n/2^{m_i})\cdot\ln{(\beta e)}}\leq 1-(1-(n/2^{m_i})\cdot\ln{(\beta e)})=(n/2^{m_i})\cdot\ln{(\beta e)}$. Here, $\beta=1.1$.

We now bound the sum of $\mathbb{P}(\mathcal{E}_{m_i})$ when $\lg{n}+2\leq {m_i}\leq c\cdot\lg{n}$:

\vspace{-3ex}
\begin{align*}
\sum_{{m_i}=\lg{n}+2}^{c\cdot\lg{n}}{(n/2^{m_i})\cdot\ln{(\beta e)}} & = n\cdot\ln{(\beta e)}\cdot\sum_{{m_i}=\lg{n}+2}^{c\cdot\lg{n}}{(1/2)^{m_i}} \\
& = n\cdot\ln{(\beta e)}\cdot\frac{(1/2)^{\lg{n}+2}\cdot(1-(1/2)^{c\cdot\lg{n}-\lg{n}-1})}{1-(1/2)} \\
& < n\cdot\ln{(\beta e)}\cdot\frac{(1/2)^{\lg{n}+2}}{1/2} = (1/2)\cdot\ln{(\beta e)} \\
& < 0.55
\end{align*}

Notice, during the execution of \CountSHCDConst, for each ${m_i}$ where $\lg{n}+2\leq {m_i}\leq c\cdot\lg{n}$, that value of $m_i$ will be used in at most one iteration (due to the properties of binary search). Since $\sum_{{m_i}=\lg{n}+2}^{c\cdot\lg{n}}{\mathbb{P}(\mathcal{E}_{m_i})}<0.55$, we know $\mathbb{P}(\bigwedge_{{m_i}=\lg{n}+2}^{c\cdot\lg{n}}{\overline{\mathcal{E}_{m_i}}})>0.45$, thus the lemma is proved.
\end{proof}

We continue to show that whenever $m_i\in[\lg{n}-1, \lg{n}+1]$, there is a constant probability that only one node will broadcast in the first slot in that iteration, thus allowing all nodes to obtain a correct estimate.

\begin{lemma}\label{lemma-CountSHCDConst-4}
During the execution of \CountSHCDConst, if $n\geq 4$, then in a iteration in which all nodes are active and $m_i$ is in range $[\lg{n}-1, \lg{n}+1]$, with probability at least $1/8$, all nodes will use $2^{m_i+1}$ as estimate of $n$ and terminate after this iteration.
\end{lemma}

\begin{proof}
Consider such an iteration, we know the probability that exactly one node broadcasts in the first time slot is at least $\min_{\lg{n}-1\leq j\leq\lg{n}+1}\{n\cdot(1/2^j)\cdot(1-1/2^j)^{n-1}\}>\min_{\lg{n}-1\leq j\leq\lg{n}+1}\{(n/2^j)\cdot 4^{-n/2^j}\}=\min\{1/8,1/4,1/4\}=1/8$.
\end{proof}

The above analysis immediately lead to Theorem \ref{thm-CountSHCDConst}.

\subsection{Omitted proofs for \CountSHCDHigh}\label{subsec-appx-omit-CountSHCDHigh}

Here prove the correctness of \CountSHCDHigh more carefully.

Let $p_{\alpha}^{(S)}$, $p_{\alpha}^{(M)}$, $p_{\alpha}^{(N)}$ be the probability that the channel is silent, \textsf{beacon}, noisy in the first slot in an iteration in which each node broadcasts with probability $\min\{1/(\alpha n),1\}$.

When $\alpha n\geq 1$, we know:

\vspace{-3ex}
\begin{align*}
p_{\alpha}^{(S)} & = \left(1-\frac{1}{\alpha n}\right)^{n} \\
p_{\alpha}^{(M)} & = n\cdot\frac{1}{\alpha n}\cdot\left(1-\frac{1}{\alpha n}\right)^{n-1} \\
p_{\alpha}^{(N)} & = 1-p_{\alpha}^{(S)}-p_{\alpha}^{(M)}
\end{align*}

Therefore, for $n\geq 10$, when $\alpha n=\hat{n}\geq 16n$:

\vspace{-3ex}
\begin{align*}
p_{\alpha}^{(S)} & = \left(1-\frac{1}{\hat{n}}\right)^{n}\geq (1.1e)^{-n/\hat{n}}\geq (1.1e)^{-1/16} > 0.93 \\
p_{\alpha}^{(M)} & = \frac{n}{\hat{n}}\cdot\left(1-\frac{1}{\hat{n}}\right)^{n-1}\leq  \frac{n}{\hat{n}}\cdot e^{-0.9n/\hat{n}}\leq \frac{1}{16}\cdot e^{-0.9/16} < 0.06\\
p_{\alpha}^{(N)} & = 1-p_{\alpha}^{(S)}-p_{\alpha}^{(M)} < 1-p_{\alpha}^{(S)} <0.07
\end{align*}

Similarly, for $n\geq 50$, when $1\leq \alpha n=\hat{n}\leq n/4$:

\vspace{-3ex}
\begin{align*}
p_{\alpha}^{(S)} & = \left(1-\frac{1}{\hat{n}}\right)^{n}\leq e^{-n/\hat{n}}\leq e^{-4} < 0.02 \\
p_{\alpha}^{(M)} & = \frac{n}{\hat{n}}\cdot\left(1-\frac{1}{\hat{n}}\right)^{n-1} \leq \frac{n}{\hat{n}}\cdot e^{-0.98n/\hat{n}} \leq 4e^{-4\cdot 0.98} < 0.08 \\
p_{\alpha}^{(N)} & = 1-p_{\alpha}^{(S)}-p_{\alpha}^{(M)} > 1-0.02-0.08 = 0.9
\end{align*}

Define $\mathcal{G}$ to be the interval $[n/4,16n]$. We divide iterations into the following types:

\begin{itemize}
	\item a \emph{good iteration}: an iteration in which the estimates used by nodes in the current iteration and the next iteration are both in $\mathcal{G}$ (notice these two estimates may be different);
	\item an \emph{improving iteration}: an iteration in which the current estimate is not in $\mathcal{G}$, but the next estimate moves towards $\mathcal{G}$;
	\item a \emph{stationary iteration}: an iteration in which the current estimate is not in $\mathcal{G}$, and the next estimate is unchanged (i.e., same with the current one);
	\item a \emph{bad iteration}: an iteration in which the current estimate is not in $\mathcal{G}$ or is the minimum or the maximum possible estimate in $\mathcal{G}$, and the next estimate moves away from $\mathcal{G}$;
\end{itemize}

During the execution of \CountSHCDHigh, let $G,I,S,B$ be the number of good, improving, stationary, and bad iterations, respectively. Let $l=a\lg{n}$ denote the number of iterations of \CountSHCDHigh, where $a$ is some sufficiently large constant.

According to our previous analysis, if in one iteration the current estimate is not in $\mathcal{G}$ or is the minimum or the maximum possible estimate in $\mathcal{G}$, then the probability that this iteration is a bad iteration is at most $\max\{0.07,0.02\}=0.07$. I.e., for each iteration, the probability that it is a bad iteration is at most $0.07$. Apply a standard Chernoff bound, we know the number of bad iterations is less than $1.01\cdot 0.07l<0.072l$, with high probability in $n$. I.e., $B<0.072l$, with high probability in $n$.

Similarly, $S<1.01\cdot\max\{0.06,0.08\}\cdot l<0.082l$, with high probability in $n$.

On the other hand, define $d=\lceil\log_{4}{(\hat{N}/(16n))}\rceil=b\lg{n}$, where $b$ is some bounded constant. According to our protocol description, we know $I\leq B+d$. Now, notice that after each iteration, the estimate is either unchanged, or is updated to an adjacent estimate (i.e., increased or decreased by a factor of four). Hence, for any estimate that is not in $\mathcal{G}$, the number of iterations in which this estimate is used is at most $\lceil B/2\rceil+\lceil I/2\rceil+S\leq 0.036l+(0.036l+d/2)+0.082l=0.154l+d/2<0.16l$.

Recall $G=l-B-I-S>l-0.072l-(0.072l+d)-0.082l=0.774l-d>0.76l$. Moreover, notice that there are at most four estimates contained in $\mathcal{G}$. Hence, during the execution of \CountSHCDHigh, there exists at least one estimate in $\mathcal{G}$ that is used in at least $0.76l/4=0.19l$ iterations.

By now, we have proved the correctness of \CountSHCDHigh.

\subsection{Omitted proofs for \CountAllnoCDa}\label{subsec-appx-omit-CountAllnoCDa}

First, we show that for any node $u$, by the end of iteration $\lg{(n_u/(b\ln{n_u}))}$, it must have not obtained its estimate yet. Here, $b\geq 1$ is a sufficiently large constant.

\begin{lemma}\label{lemma-CountAllnoCDa-1}
During the execution of \CountAllnoCDa, for a node $u$, if $n_u$ is sufficiently large, then for each iteration $i$ where $1\leq i\leq \lg{(n_u/(b\ln{n_u}))}$, $u$ will not obtain its estimate after iteration $i$, with high probability in $n_u$. Here, $b\geq 1$ is a sufficiently large constant.
\end{lemma}

\begin{proof}
Consider a time slot in which $u$ decides to listen. In expectation, $n_u/2$ neighbors of $u$ will choose to broadcast. A standard Chernoff bound implies with probability at least $1-e^{-n_u/144}$, at least $n_u/3$ neighbors of $u$ will choose to broadcast. For sufficiently large $n_u$, this means at least $n_u/3$ neighbors of $u$ will choose to broadcast, with high probability in $n_u$. Take a union bound over the $O(\lg^2{n_u})$ time slots during the first $\lg{(n_u/(b\ln{n_u}))}$ iterations, we know during these iterations, whenever $u$ chooses to listen, at least $n_u/3$ neighbors of $u$ will broadcast, with high probability in $n_u$.

Now, consider a time slot in iterations $i$ in which $u$ decides to listen, where $1\leq i\leq \lg{(n_u/(b\ln{n_u}))}$. Assume in this slot $x$ neighbors of $u$ broadcast, where $n_u/3\leq x\leq n_u$. We know the probability that $u$ hears a \textsf{beacon} message is $x\cdot(1/2^i)\cdot(1-1/2^i)^{x-1}\leq x\cdot(1/2^i)\cdot(1-1/2^i)^{x/2}\leq n_u\cdot(1/2)\cdot(1-(b\ln{n_u})/n_u)^{n_u/6}\leq (n_u/2)\cdot e^{-(b\ln{n_u}/n_u)\cdot(n_u/6)}=(n_u/2)\cdot e^{-(b/6)\ln{n_u}}={n_u}^{-\Theta(1)}$. Take a union bound over all the $O(\lg^2{n_u})$ time slots during the first $\lg{(n_u/(b\ln{n_u}))}$ iterations, we know the probability that $u$ hears a \textsf{beacon} message when it chooses to listen is at most ${n_u}^{-\Theta(1)}$. This proves the lemma.
\end{proof}

We then consider iterations $\lg{(n_u/(b\ln{n_u}))}$ to $\lg{n_u}-4$, and show that $u$ will not decide its estimate during these iterations as well.

\begin{lemma}\label{lemma-CountAllnoCDa-2}
During the execution of \CountAllnoCDa, for a node $u$, if $n_u$ is sufficiently large and $u$ has not obtained its estimate by the end of iteration $\lg{(n_u/(b\ln{n_u}))}-1$, then for each iteration $i$ where $\lg{(n_u/(b\ln{n_u}))}\leq i\leq\lg{n_u}-4$, $u$ will not obtain its estimate after iteration $i$, with high probability in $n_u$. Here, $b\geq 1$ is a sufficiently large constant.
\end{lemma}

\begin{proof}
Firstly, it is easy to prove that during iteration $\lg{(n_u/(b\ln{n_u}))}$ to $\lg{n_u}-4$, whenever $u$ chooses to listen, at least $0.41n_u$ and at most $0.6n_u$ neighbors of $u$ will broadcast, with high probability in $n_u$. (This can be proved by a similar argument shown in the first paragraph of the proof for Lemma \ref{lemma-CountAllnoCDa-1}.)

Assume the $i$\textsuperscript{th} iteration contains $ai$ time slots, where $a$ is a sufficiently large constant. We know in iteration $i$, $u$ will choose to listen in $ai/2$ time slots, in expectation. Apply a Chernoff bound and a union bound, we know in each of the iterations from $\lg{(n_u/(b\ln{n_u}))}$ to $\lg{n_u}-4$, $u$ will choose to listen in at least $0.99a/2\cdot i$ time slots, and at most $1.01a/2\cdot i$ time slots, with high probability in $n_u$.

Assume $u$ chooses to listen in $l_i$ time slots in iteration $i$, and $n_j$ neighbors of $u$ decide to broadcast in the $j$\textsuperscript{th} slot of these $l_i$ slots. Let $X_i$ be a random variable denoting the number of slots in which $u$ hears \textsf{beacon} during the $l_i$ listening slots in iteration $i$. Let $\mu_i=\mathbb{E}(X_i)=\sum_{j=1}^{l_i}{(n_j\cdot(1/2^i)\cdot(1-1/2^i)^{n_j-1})}$. When $\lg{(n_u/(b\ln{n_u}))}\leq i\leq\lg{n_u}-4$, we know: $\mu_i\leq \sum_{j=1}^{l_i}{(n_j\cdot(16/n_u)\cdot(1-16/n_u)^{n_j-1})}\leq l_i\cdot 0.6n_u\cdot(16/n_u)\cdot(1-16/n_u)^{0.41n_u-1}\leq l_i\cdot 0.6n_u\cdot(16/n_u)\cdot(1-16/n_u)^{0.4n_u}\leq 9.6l_i\cdot e^{-(16/n_u)\cdot 0.4n_u}=9.6e^{-6.4}\cdot l_i$. Notice, the first inequality holds since $\mu_i$ is non-decreasing when $\lg{(n_u/(b\ln{n_u}))}\leq i\leq\lg{n_u}-4$.

Define $t_i=28.8e^{-6.4}\cdot l_i$ and $\delta_i=t_i/\mu_i$. Due to Claim \ref{claim-chernoff-bound-variant}, we know $\mathbb{P}(X_i\geq t_i)<(e\mu_i/t_i)^{t_i}=(e/t_i)^{t_i}\cdot{\mu_i}^{t_i}\leq (e/t_i)^{t_i}\cdot{(9.6e^{-6.4}\cdot l_i)}^{t_i}=(e/3)^{t_i}\leq (e/3)^{28.8e^{-6.4}\cdot\Theta(\lg{(n_u/(b\ln{n_u}))})}=(e/3)^{\Theta(\lg{n_u})}=1/{n_u}^{\Theta(1)}$. Take a union bound over iterations $\lg{(n_u/(b\ln{n_u}))}$ to $\lg{n_u}-4$, we know with high probability in $n_u$, node $u$ will hear \textsf{beacon} messages in at most $28.8e^{-6.4}<1/10$ fraction of listening slots during each of iterations $\lg{(n_u/(b\ln{n_u}))}$ to $\lg{n_u}-4$. According to our protocol, this is not enough for $u$ to decide its estimate.
\end{proof}

The last key technical lemma for \CountAllnoCDa states that $u$ must have decided its estimate by the end of iteration $\lg{n_u}-1$.

\begin{lemma}\label{lemma-CountAllnoCDa-3}
During the execution of \CountAllnoCDa, for a given node $u$, if $n_u$ is sufficiently large, then by the end of iteration $\lg{n_u}-1$, it must have obtained its estimate, with high probability in $n_u$.
\end{lemma}

\begin{proof}
Without loss of generality, assume $u$ has not decided its estimate by the end of iteration $\lg{n_u}-2$. In iteration $\lg{n_u}-1$, it is easy to show that whenever $u$ chooses to listen, at least $0.4n_u$ and at most $0.6n_u$ neighbors of $u$ will broadcast, with high probability in $n_u$. (This can be proved by a similar argument shown in the first paragraph of the proof for Lemma \ref{lemma-CountAllnoCDa-1}.) Moreover, assuming iteration $\lg{n_u}-1$ contains $a(\lg{n_u}-1)$ time slots where $a$ is a sufficiently large constant, then a Chernoff bound implies $u$ will choose to listen in at least $0.9a(\lg{n_u}-1)/2$ time slots, and at most $1.1a(\lg{n_u}-1)/2$ time slots, with high probability in $n_u$.

Assume $u$ chooses to listen in $l=\Theta(\lg{n_u})$ time slots in iteration $\lg{n_u}-1$, and $n_j$ neighbors of $u$ decide to broadcast in the $j$\textsuperscript{th} slot of these $l$ slots. Let $X$ be a random variable denoting the number of slots in which $u$ hears a \textsf{beacon} message during the $l$ listening slots in iteration $\lg{n_u}-1$. We know $\mathbb{E}(X)=\sum_{j=1}^{l}{(n_j\cdot(2/n_u)\cdot(1-2/n_u)^{n_j-1})}>l\cdot 0.4n_u\cdot(2/n_u)\cdot(1-2/n_u)^{0.6n_u}=0.8l\cdot (1-2/n_u)^{0.6n_u}>0.8l\cdot 4^{-(2/n_u)\cdot 0.6n_u}=0.8l/4^{1.2}>0.15l$. Apply a standard Chernoff bound we know $X$ will be at least $l/10$, with probability at least $1-e^{-\Theta(l)}=1-1/{n_u}^{\Theta(1)}$. Hence, $u$ will decide its estimate after iteration $\lg{n_u}-1$ if it has not done so already, with high probability in $n_u$.
\end{proof}

The above lemmas immediately give the correctness proof of \CountAllnoCDa.

\subsection{Omitted description and analysis for \CountAllnoCDb}\label{subsec-appx-omit-CountAllnoCDb}

We first give a specification for each iteration. Each iteration contains $\lg{N}$ sub-iterations, each of which has $\Theta(\lg{N})$ time slots. In each time slot in an iteration, a listener will simply listen. On the other hand, in each time slot in the $i$\textsuperscript{th} sub-iteration, each broadcaster will broadcast a \textsf{beacon} message with probability $1/2^i$. Moreover, for a listener $u$, if by the end of sub-iteration $i$, for the first time since the beginning of this iteration, it has heard \textsf{beacon} messages in at least $1/40$ fraction of slots within this sub-iteration, then $u$ will use $2^{i+2}$ as its estimate for the number of broadcasting neighbors (in this iteration).

We still need to prove Lemma \ref{lemma-CountAllnoCDb-1}.

\begin{proof}[Proof of Lemma \ref{lemma-CountAllnoCDb-1}.]
In sub-iteration $i$, in a slot, the probability that $u$ will hear a \textsf{beacon} message is $m\cdot(1/2^{i})\cdot(1-1/2^{i})^{m-1}$. Hence, if the length of a sub-iteration is $a\lg{N}$ where $a$ is some sufficiently large constant, then in expectation, the listener will hear $a\cdot m\cdot(1/2^{i})\cdot(1-1/2^{i})^{m-1}\cdot\lg{N}$ messages in round $i$. Let random variable $X_i$ denote this number, we have $\mathbb{E}(X_i)=a\cdot m\cdot(1/2^{i})\cdot(1-1/2^{i})^{m-1}\cdot\lg{N}$.

Define function $f(x)=a\cdot m\cdot x\cdot(1-x)^{m-1}\cdot\lg{N}$. (Here, $x$ is effectively representing $1/2^{i}$.) The first order derivative of $f(x)$ is $f'(x)=(a\cdot m\cdot(1-x)^{m-2}\cdot\lg{N})\cdot(1-x\cdot m)$. Hence, $f(x)$ is maximized when $xm=1$. In turn, this suggests, when $2^{i}=m$, the value of $\mathbb{E}(X_i)$ will be maximized.

Therefore, when $1\leq 2^{i}\leq m/8$, we know $\mathbb{E}(X_i)\leq a\cdot m\cdot(8/m)\cdot(1-8/m)^{m-1}\cdot\lg{N}=8a\lg{N}\cdot(1-8/m)^{m-1}\leq 8a\lg{N}\cdot e^{-8(m-1)/m}\leq (a\lg{N})\cdot(8e^{-8(7m/8)/m})\leq(a\lg{N})\cdot(8e^{-7})$. Since nodes' choices made among different slots are independent, apply a Chernoff bound and we know, when $2^{i}\leq m/8$, in one sub-iteration, the maximum fraction of time slots in which the listener can hear a message is $(1+\delta)\cdot(8e^{-7})$, w.h.p.\ in $N$. Here, when $a$ is sufficiently large, $0<\delta<1$ is an arbitrarily small constant.

On the other hand, when $m\geq 2^{i}\geq m/2\geq 2$, we know $\mathbb{E}(X_i)\geq a\cdot m\cdot(2/m)\cdot(1-2/m)^{m-1}\cdot\lg{N}=2a\lg{N}\cdot(1-2/m)^{m-1}\geq 2a\lg{N}\cdot e^{-4(m-1)/m}\geq (a\lg{N})\cdot(2e^{-4})$. Since nodes' choices made among different slots are independent, apply a Chernoff bound and we know, when $m\geq 2^{i}\geq m/2$, in one sub-iteration, the minimum fraction of time slots in which the listener can hear a message is $(1-\delta)\cdot(2e^{-4})$, w.h.p.\ in $N$. Again, when $a$ is sufficiently large, $0<\delta<1$ is an arbitrarily small constant.

Recall our algorithm asks $u$ to decide its estimate for an iteration in the first sub-iteration in which the fraction of clear message slots is at least $1/40$. With a union bound over the $\lg{N}$ sub-iterations, we know $u$ will not decide its estimate when $2^{i}\leq m/8$, and must have decided its estimate when $2^{i}\geq m$. (Since $8e^{-7}<1/40<2e^{-4}$.) Hence, we know the listener will obtain an estimate in range $[m,4m]$, with high probability in $N$.
\end{proof}

\subsection{Omitted proofs for \CountAllCDa}\label{subsec-appx-omit-CountAllCDa}

\begin{proof}[Proof sketch of Theorem \ref{thm-CountAllCDa}.]
Consider a node $u$. According to the protocol, it is easy to see all neighbors of $u$ will not terminate until $u$ has obtained its estimate for $n_u$. Thus, by the analysis of \CountAllnoCDa, we can conclude $u$ will obtained an estimate of $n_u$ in range $[n_u,4n_u]$ within $O(\lg^2{n_u})$ time slots, with high probability in $n_u$.

We now analyze when will $u$ terminate. We know by the end of part one of iteration $\lg{n_u}-1$, node $u$ must have obtained its estimate of $n_u$, with high probability in $n_u$. Assume by the end of part one of iteration $\lg{n_u}-1$, node $u$ indeed has obtained its estimate of $n_u$. Consider a neighbor $v$ of $u$. If $n_v\leq n_u$, then by the end of part one of iteration $\lg{n_u}-1$, we know with probability at least $1-1/{n_v}$, node $v$ has already obtained its estimate of $n_v$, and will not broadcast \textsf{continue} to try to stop $u$ from terminating. On the other hand, if $n_v> n_u$, then by the end of part one of iteration $\lg{n_v}-1$, we know with probability at least $1-1/{n_v}$, node $v$ has already obtained its estimate of $n_v$, and will not broadcast \textsf{continue} to try to stop $u$ from terminating. To sum up, by the end of part one of iteration $\max\{\lg{n_u},\lg{n_v}\}$, node $v$ will not broadcast \textsf{continue} to try to stop $u$ from terminating, with probability at least $1-1/{n_v}$. Take a union bound over all neighbors of $u$, we know by the end of part one of iteration $\max\{\lg{n_u},\max_{v\in\Gamma_u}\{\lg{n_v}\}\}$, no neighbor of $u$ will try to stop $u$ from terminating, with probability at least $1-\sum_{v\in\Gamma_u\cup\{u\}}(1/{n_v})$.
\end{proof}

\section{Upper Bounds for Multi-hop with Designated Node Counting}\label{sec-appx-omit-alg-center}

\subsection{No collision detection}\label{subsec-appx-omit-alg-center-noCD}

\subparagraph*{Constant probability of success.} \CountCenternoCDConst contains multiple iterations, each of which contains $l+1$ time slots. Here, $l$ is a constant, and can be adjusted to achieve desirable correctness guarantees. In the $i$\textsuperscript{th} iteration, nodes assume $n_w\approx 2^{i}$ and use the first $l$ time slots to determine the accuracy of the estimate. (Recall $w$ denotes the designated node.) More specifically, $w$ will listen in all these $l$ time slots; while each neighbor of $w$ will broadcast a \textsf{beacon} message with probability $1/2^{i}$ in each of these time slots. After these $l$ time slots, for the designated node $w$, if the fraction of time slots in which it heard \textsf{beacon} messages is at least $1/2e$, then it will broadcast a \textsf{stop} message in the last time slot to inform all nodes to terminate, and itself will use $2^{i+2}$ as the estimate of $n_w$. Otherwise, $w$ will keep silent in the last time slot to instruct all nodes to continue into the next iteration.

Similar to \CountSHnoCDConst, in order to prove the correctness of \CountCenternoCDConst, we first show that by the end of iteration $\lg{n_w}-3$, all nodes are still active, with at least some constant probability (that can be arbitrarily close to one). The high level strategy for proving this claim is similar to that of proving Lemma \ref{lemma-CountSHnoCDConst-1}, but the actual proof is more involved.

\begin{lemma}\label{lemma-CountCenternoCDConst-1}
During the execution of \CountCenternoCDConst, if $n_w\geq 10$ and $l$ is sufficiently large, then by the end of iteration $\lg{n_w}-3$, designated node $w$ and all of its neighbors are still active, with probability at least $1-\epsilon$. Here, $0<\epsilon<1$ is an arbitrarily small constant.
\end{lemma}

\begin{proof}
Assume $w$ and all of its neighbors are still active at the beginning of iteration $i$ where $i\leq\lg{n_w}-3$. Let $X_i$ be a random variable denoting the number of slots in which $w$ hears a \textsf{beacon} message during the first $l$ time slots in iteration $i$. Let $\mu_i=\mathbb{E}(X_i)$. We know $\mu_i=l\cdot n_w\cdot(1/2^i)\cdot(1-1/2^i)^{n_w-1}\leq l\cdot n_w\cdot(8/n_w)\cdot(1-8/n_w)^{n_w-1}\leq 8l\cdot(1-8/n_w)^{0.9n_w}\leq 8l\cdot e^{-(8/n_w)\cdot 0.9n_w}=8l\cdot e^{-7.2}$. Notice, the first inequality holds since $\mu_i$ is non-decreasing during iterations $1\leq i\leq\lg{n_w}-3$.

Define $t=8lb\cdot e^{-7.2}$ and $\delta_i=t/\mu_i$, where $b$ is some positive constant to be fixed. Since $\mu_i\leq l\cdot n_w\cdot(1/2^i)\cdot(1-1/2^i)^{0.9n_w}\leq l\cdot n_w\cdot(1/2^i)\cdot e^{-0.9n_w/2^i}$, we know $\delta_i\geq (8b\cdot e^{-7.2}/n_w)\cdot 2^i\cdot e^{0.9n_w/2^i}$. Thus, due to Claim \ref{claim-chernoff-bound-variant}, we know $\mathbb{P}(X_i\geq t)<(e^{8.2} n_w/8b)^t\cdot(e^{-0.9n_w/2^i}/2^i)^t$.

Therefore, we can conclude:

\vspace{-3ex}
\begin{align*}
\sum_{i=1}^{\lg{n_w}-3}{\mathbb{P}(X_i\geq t)} & < \sum_{i=1}^{\lg{\frac{n_w}{8}}}{\left(\frac{e^{8.2}\cdot n_w}{8b}\right)^t\cdot\left(\frac{1}{2^i}\cdot e^{-0.9n_w/2^i}\right)^t} \\
& < \sum_{j=8}^{n_w}{\left(\frac{e^{8.2}\cdot n_w}{8b}\right)^t\cdot\left(\frac{j}{n_w}\cdot e^{-0.9j}\right)^t} = \left(\frac{e^{8.2}}{8b}\right)^t \cdot \sum_{j=8}^{n_w}{\left(j\cdot e^{-0.9j}\right)^t} \\
& \leq \left(\frac{e^{8.2}}{8b}\right)^t \cdot \left(\sum_{j=8}^{n_w}{j\cdot e^{-0.9j}}\right)^t
\end{align*}

where the last inequality clearly follows if $t$ is some positive integer (and it will be).

Define $S=\sum_{j=8}^{n_w}{j\cdot e^{-0.9j}}$, we know:

\vspace{-3ex}
\begin{align*}
S-e^{-0.9}\cdot S & = (1-e^{-0.9})S = \left(\sum_{i=8}^{n_w}{(e^{-0.9})^i}\right) + 7e^{-7.2} - n_w\cdot(e^{-0.9})^{n_w+1} \\
& < \left(\sum_{i=8}^{n_w}{(e^{-0.9})^i}\right) + 7e^{-7.2} = \frac{e^{-7.2}\cdot(1-(e^{-0.9})^{n_w-7})}{1-e^{-0.9}} + 7e^{-7.2} \\
& < \frac{e^{-7.2}}{1-e^{-0.9}} + 7e^{-7.2}
\end{align*}

Hence, $S<1/90$, which implies $\sum_{i=1}^{\lg{n_w}-3}{\mathbb{P}(X_i\geq t)}<(e^{8.2}/8b)^t\cdot(1/90)^t=(e^{8.2}/720b)^t$. Now, fix $b=10$, and choose $l$ such that $t=8lb\cdot e^{-7.2}=80e^{-7.2}\cdot l$ is a positive integer. As a result, $\sum_{i=1}^{\lg{n_w}-3}{\mathbb{P}(X_i\geq 80e^{-7.2}\cdot l)}<(e^{8.2}/720b)^t=(e^{8.2}/7200)^{80e^{-7.2}\cdot l}$. I.e., for sufficiently large constant $l$, this sum is an arbitrarily small positive constant. This implies, with constant probability that can be arbitrarily close to one, during each iteration $i\leq \lg{n_w}-3$, the designated node $w$ will hear \textsf{beacon} messages in at most $80e^{-7.2}$ fraction of time slots. Recall according to our protocol, nodes will only terminate if the designated node hears \textsf{beacon} messages in at least $1/2e$ fraction of time slots. Therefore, the lemma is proved.
\end{proof}

Next, we show that all nodes must have terminated by the end of iteration $\lg{n_w}$, with at least some constant probability (that can be arbitrarily close to one).

\begin{lemma}\label{lemma-CountCenternoCDConst-2}
During the execution of \CountCenternoCDConst, assume all nodes are active at the beginning of iteration $\lg{n_w}$. In such case, if $n_w\geq 2$ and $l$ is sufficiently large, then by the end of iteration $\lg{n_w}$, designated node $w$ and all of its neighbors must have terminated, with probability at least $1-\epsilon$. Here, $0<\epsilon<1$ is an arbitrarily small constant.
\end{lemma}

\begin{proof}
Since we assume all nodes are active at the beginning of iteration $\lg{n_w}$, we know in each time slot in this iteration, the probability that exactly one neighbor of $w$ broadcasts is $n_w\cdot(1/n_w)\cdot(1-1/n_w)^{n_w-1}\geq (1-1/n_w)^{n_w}\geq 1/4$. Thus, in expectation, the number of time slots in which $w$ hears a \textsf{beacon} message is at least $l/4$. Recall according to our protocol, nodes will terminate if $w$ hears \textsf{beacon} message in at least $l/2e$ time slots. Since each time slot is independent, apply a Chernoff bound and we know $w$ will hear \textsf{beacon} message in at least $l/2e$ time slots, with probability at least $1-e^{-\Theta(l)}$.
\end{proof}

The following theorem immediately follows from the above two lemmas.

\begin{theorem}\label{thm-CountCenternoCDConst}
The \CountCenternoCDConst approximate neighbor counting algorithm guarantees the following properties when executed in a multi-hop network with no collision detection:
(a) the designated node $w$ and all its neighbors terminate simultaneously;
and (b) with probability at least $1-\epsilon$, the designated node $w$ obtains the estimate of $n_w$ in range $[n_w,4n_w]$ within $O(\lg{n_w})$ time slots. Here, $0<\epsilon<1$ is an arbitrary constant.
\end{theorem}

\subparagraph*{High probability of success.} If we want to guarantee the estimate of $n_w$ is correct with high probability in $n_w$, only small adjustments to \CountCenternoCDConst are needed: for every iteration $i$, increase the number of time slots in which each neighbor broadcasts from $l=\Theta(1)$ to $\Theta(i)$.

To prove the correctness of \CountCenternoCDHigh, we follow the strategies that are used for the analysis of \CountSHnoCDHigh: first, show that during iterations one to $\lg{(n_w/(a\ln{n_w}))}$ no node will terminate as the designated node can never hear a \textsf{beacon} message (since the estimate is too small); then, show that during iterations $\lg{(n_w/(a\ln{n_w}))}$ to $\lg{n_w}-\Theta(1)$ (such as $\lg{n_w}-3$) still no node will terminate as the fraction of time slots in which the designated node has heard a \textsf{beacon} message will not reach the threshold; and finally, show that by the end of iteration $\lg{n_w}$, all nodes must have terminated since the designated node must have heard sufficient \textsf{beacon} messages in this iteration.

We omit the detailed proofs for \CountCenternoCDHigh, as they are almost identical to the ones for \CountSHnoCDHigh.

\subsection{Collision detection is available}

In such case, we can use algorithms that are very similar to \CountSHCDConst or \CountSHCDHigh to achieve our goals. More specifically, we will briefly describe two algorithms such that one---called \CountCenterCDConst---can solve the considered problem with constant probability, and the other---called \CountCenterCDHigh---can solve the considered problem with high probability in $n_w$ (recall $w$ is the designated node and $n_w$ is the number of neighbors it has). We will not give detailed correctness proofs for these two algorithms since these proofs are almost identical to the ones for \CountSHCDConst and \CountSHCDHigh.

Before describing \CountCenterCDConst and \CountCenterCDHigh, we first discuss how to adopt \EstUpperSH to the multi-hop designated node counting scenario. We call this variant of \EstUpperSH as \EstUpperCenter. \EstUpperCenter contains multiple iterations. In the $i$\textsuperscript{th} iteration, the designated node $w$ assumes $\lg{n_w}=2^i$, and verifies the accuracy of the estimate. In case the estimate is accurate, $w$ informs all neighbors and \EstUpperCenter is done. Otherwise, all nodes proceed into the next iteration. More specifically, the $i$\textsuperscript{th} iteration contains two slots. In the first slot, each neighbor of $w$ broadcasts a \textsf{beacon} message with probability $1/2^{2^i}$, and $w$ simply listens. If $w$ hears a message or silence, then in the second slot $w$ broadcasts a \textsf{stop} message to inform all nodes to terminate, with $2^{i+1}$ being the estimate of $\lg{n_w}$. Otherwise, $w$ remains silent in the second slot and all nodes proceed into the next iteration.

\textbf{\CountCenterCDConst} contains multiple iterations, each of which contains two time slots. In each iteration $i$, all nodes have a lower bound $a_i$ and an upper bound $b_i$, and will test whether the median $m_i=\lfloor (a_i+b_i)/2\rfloor$ is close to $\lg{n_w}$. (Initially, $a_1$ is set to one, and $b_1$ is set to the estimate of $\lg{n_w}$ returned by \EstUpperCenter.) To achieve this goal, in the first slot in iteration $i$, each neighbor of $w$ will choose to broadcast a \textsf{beacon} message with probability $1/2^{m_i}$, and the designated node $w$ will simply listen. If $w$ finds no node broadcasts in the first slot, it will broadcast an \textsf{over-est} message in the second slot to inform all nodes to set $b_{i+1}$ to $m_i-1$; and if $w$ finds multiple nodes broadcast in the first slot, it will broadcast an \textsf{under-est} message in the second slot to inform all nodes to set $a_{i+1}$ to $m_i+1$. Moreover, if $w$ finds exactly one node broadcasts in the first slot, it will broadcast a \textsf{stop} message in the second slot to inform all nodes to terminate, with $2^{m_i+1}$ being the estimate of $n_w$. Finally, if in an iteration $a_i>b_i$, then all nodes will simply abort without obtaining an estimate of $n_w$.

\textbf{\CountCenterCDHigh} contains $\Theta(\lg{\hat{N}_w})$ iterations, each of which containing two time slots. Here, $\hat{N}_w$ is the polynomial upper bound of $n_w$ returned by \EstUpperCenter. In each iteration, all nodes maintain a current estimate on $n_w$ which is denoted by $\hat{n}_w$. (Initially, $\hat{n}_w$ is set to $\hat{N}_w$.) In the first slot in an iteration, each neighbor of $w$ will broadcast a \textsf{beacon} message with probability $1/\hat{n}_w$, and $w$ will simply listen. If $w$ hears silence, then it will instruct all nodes to decrease $\hat{n}_w$ by a factor of four in the next iteration; if $w$ hears noise, then it will instruct all nodes to increase $\hat{n}_w$ by a factor of four in the next iteration; and if $w$ hears \textsf{beacon} from some neighbor, then it will instruct all nodes to keep $\hat{n}_w$ unchanged in the next iteration. After these $\Theta(\lg{\hat{N}_w})$ iterations, $w$ will use $4\tilde{n}_w$ to be the final estimate of $n_w$, where $\tilde{n}_w$ is the most frequent estimate used by the nodes during the $\Theta(\lg{\hat{N}_w})$ iterations.

\end{document}